\documentclass[11pt]{article}
\pdfoutput=1 

\usepackage{jheppub} 

\usepackage[T1]{fontenc} 

\usepackage{euscript,amsmath,amsthm,amsfonts,amssymb}
\usepackage[arrow,frame,matrix]{xy}
\usepackage{subcaption}

\newtheorem{theorem}{Theorem}[section]
\newtheorem{lemma}[theorem]{Lemma}
\newtheorem{conjecture}[theorem]{Conjecture}

\newcommand{\ket}[1]{|{#1}\rangle}

\newcommand{\bra}[1]{\langle{#1}|}

\newcommand{\R}{\mathbf{R}}
\newcommand{\Z}{\mathbf{Z}}

\DeclareMathOperator{\tr}{tr}

\DeclareMathOperator{\area}{area}

\title{\boldmath Bit threads and holographic entanglement}

\author[a]{Michael Freedman}
\author[b]{and Matthew Headrick}

\affiliation[a]{Station Q, Microsoft Research, Santa Barbara, California 93106, USA and Department of Mathematics, University of California, Santa Barbara, California 93106 USA}
\affiliation[b]{Martin Fisher School of Physics, Brandeis University, Waltham, Massachusetts 02453, USA}

\abstract{The Ryu-Takayanagi (RT) formula relates the entanglement entropy of a region in a holographic theory to the area of a corresponding bulk minimal surface. Using the max flow-min cut principle, a theorem from network theory, we rewrite the RT formula in a way that does not make reference to the minimal surface. Instead, we invoke the notion of a ``flow'', defined as a divergenceless norm-bounded vector field, or equivalently a set of Planck-thickness ``bit threads''. The entanglement entropy of a boundary region is given by the maximum flux out of it of any flow, or equivalently the maximum number of bit threads that can emanate from it. The threads thus represent entanglement between points on the boundary, and naturally implement the holographic principle. As we explain, this new picture clarifies several conceptual puzzles surrounding the RT formula. We give flow-based proofs of strong subadditivity and related properties; unlike the ones based on minimal surfaces, these proofs correspond in a transparent manner to the properties' information-theoretic meanings. We also briefly discuss certain technical advantages that the flows offer over minimal surfaces. In a mathematical appendix, we review the max flow-min cut theorem on networks and on Riemannian manifolds, and prove in the network case that the set of max flows varies Lipshitz continuously in the network parameters.
}

\preprint{BRX-TH-6302, NSF-KITP-16-051}

\begin{document} 
\maketitle
\flushbottom

\section{Introduction}
\label{sec:intro}

\begin{figure}[tbp]
\centering
\includegraphics[width=0.5\textwidth]{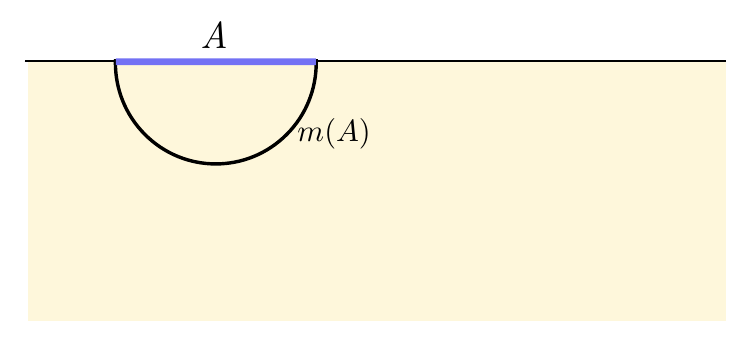}
\caption{\label{fig:RT}
According to the Ryu-Takayanagi formula, \eqref{RT}, the entanglement entropy $S(A)$ of a given boundary spatial region is given by the area of a corresponding bulk minimal surface $m(A)$.
}
\end{figure}

The Ryu-Takayanagi entanglement entropy formula \cite{Ryu:2006bv,Ryu:2006ef} is by now a firmly established entry in the holographic dictionary. This formula, which applies when the bulk is static and governed by classical Einstein gravity,\footnote{We will restrict our attention in the bulk of this paper to the regime of applicability of the RT formula. In the last section, we will briefly discuss its covariant generalization \cite{Hubeny:2007xt} as well as stringy and quantum corrections.} gives the EE of an arbitrary spatial region $A$ in terms of the area of $m(A)$, the minimal bulk surface homologous to $A$ (fig.\ \ref{fig:RT}):
\begin{equation}\label{RT}
S(A) = \frac1{4G_{\rm N}}\area(m(A))\,.
\end{equation}
In addition to being calculationally useful, this beautiful formula is widely believed to contain some deep---but still hidden---conceptual message about the nature of quantum gravity and the emergence of spacetime.

In trying to decode the conceptual implications of the RT formula, it is natural to wonder how one should think about the minimal surface $m(A)$, to which the formula seems to assign a special status. A naive interpretation is that the bits encoding the microstate of $A$ somehow ``live on'' the minimal surface $m(A)$, at a density of one bit per four Planck areas.\footnote{\label{bits}Actually, $\ln2$ bits. As our aims are mostly conceptual, for simplicity of presentation we will consistently misuse ``bit'' to mean ``$\ln 2$ bits'' (sometimes called a ``nat'').} A similar interpretation can be given to the Bekenstein-Hawking black-hole entropy formula (which to a certain extent is a special case of the RT formula). However, whereas the location of the black-hole horizon is fixed by the causal structure of the spacetime, the minimal surface $m(A)$ depends on the arbitrary choice of boundary region $A$, and this freedom reveals several problems with the above interpretation.

\begin{figure}[tbp]
\centering
\includegraphics[width=0.9\textwidth]{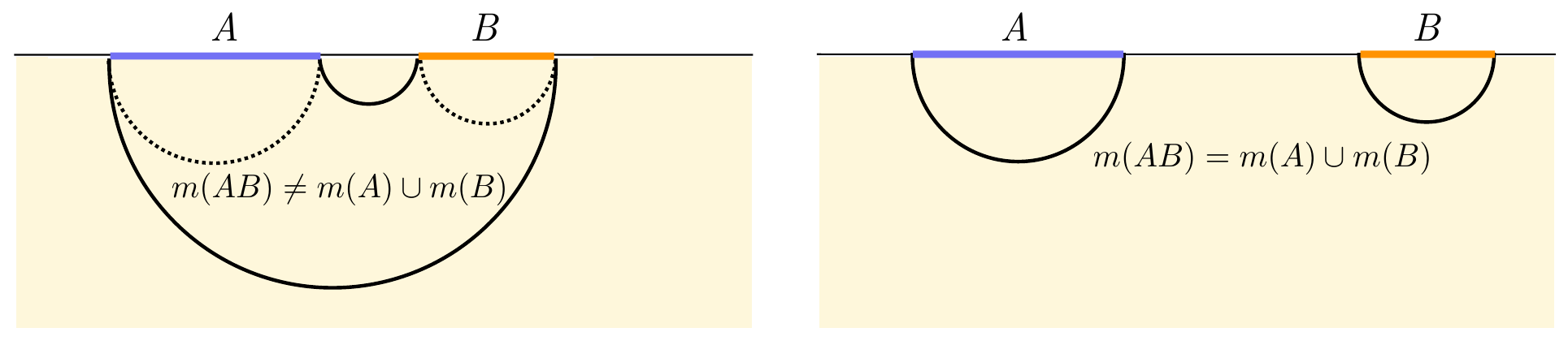}
\caption{\label{fig:transition}
The minimal surface for the union of two separated regions undergoes a transition as a function of their separation between connecting them at small separation (left) and equalling the union of their respective minimal surfaces at large separation (right).}
\end{figure}
First, the minimal surface can jump under continuous deformations of $A$ \cite{Hirata:2006jx,Nishioka:2006gr,Klebanov:2007ws,Headrick:2010zt}, suggesting that the bits strangely jump from one place to another.\footnote{One might wonder whether such a jump, which is due to competing minimal surfaces, reflects a jump in the type of microstate represented in the reduced density matrix $\rho_A$, in other words a first-order phase transition between competing macrostates. However, this seems unlikely since, according to the RT formula, the entropy is given by the \emph{least}-area surface, whereas in a conventional phase transition, it is the macrostate with the \emph{largest} entropy that dominates (in the microcanonical ensemble). See \cite{Headrick:2013zda} for further discussion of this issue.} Consider for example the classic case of two separated regions $A$, $B$. The minimal surface $m(AB)$ for their union typically connects them at sufficiently small separation; as the separation is increased, however, it jumps to being the union of their respective minimal surfaces $m(A)\cup m(B)$ (fig.\ \ref{fig:transition}).

\begin{figure}[tbp]
\centering
\includegraphics[width=0.7\textwidth]{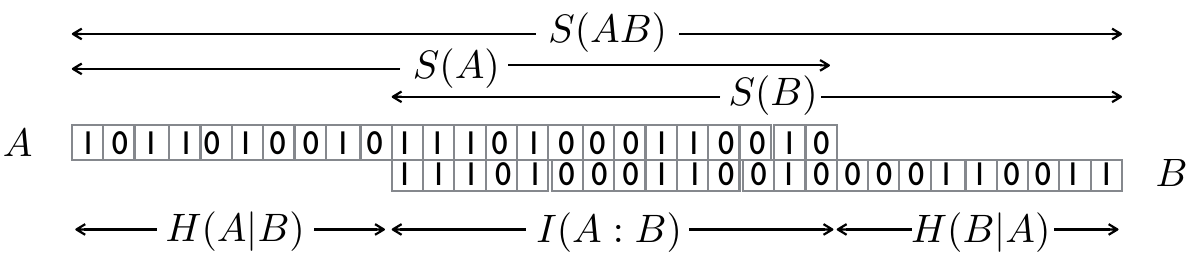}
\caption{\label{fig:MIclassical}
Schematic illustration of the mutual information and conditional entropy in a classical system: Two correlated systems $A$ and $B$ can be encoded into $S(A)$ and $S(B)$ bits respectively, such that $I(A:B)$ bits of each are perfectly correlated, $H(A|B)$ bits of $A$ are uncorrelated with those of $B$, and $H(B|A)$ bits of $B$ are uncorrelated with those of $A$.
}
\end{figure}

\begin{figure}[tbp]
\centering
\includegraphics[width=0.7\textwidth]{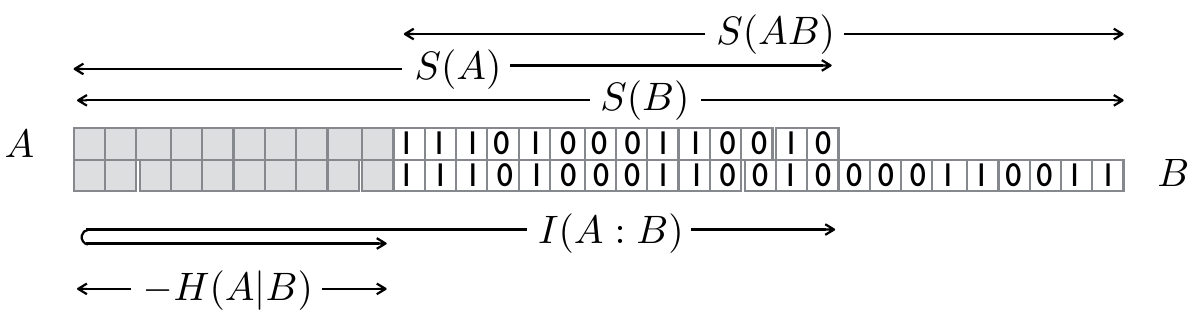}
\caption{\label{fig:MIquantum}
The shaded boxes indicate pairs of bits that are maximally entangled between $A$ and $B$. Each such EPR pair contributes 1 to $S(A)$ and $S(B)$, 0 to $S(AB)$, 2 to $I(A:B)$, and $-1$ to $H(A|B)$. In this case, $H(A|B)$ is negative.
}
\end{figure}

Related to this, it seems mysterious why the conditional entropy,
\begin{equation}
H(A|B):=S(AB)-S(B)\,,
\end{equation}
which measures the expected entropy of $A$ conditioned on knowing the state of $B$, and the mutual information
\begin{equation}
I(A:B):=S(A)+S(B)-S(AB)\,,
\end{equation}
which measures the amount of correlation between $A$ and $B$, should be given by a difference of areas of surfaces that may pass through different parts of the spacetime (as in the left side of fig.\ \ref{fig:transition}). To put this question into context, let us briefly and heuristically recall why these particular linear combinations of entropies have special information-theoretic meanings, starting with the classical case. The state of $AB$ can be encoded in a compressed form, such that it is represented by $S(AB)$ bits, the state of $A$ by $S(A)$ bits, and the state of $B$ by $S(B)$ bits. Then $I(A:B)$ is clearly the number of bits that are shared by $A$ and $B$, while $H(A|B)$ is the number that appear only in $A$ and $H(B|A)$ the number that appear only in $B$ (fig.\ \ref{fig:MIclassical}). In the quantum case, a bit\footnote{Strictly speaking, a qubit. For simplicity, in this paper we will apply the term ``bit'' uniformly to the classical and quantum cases.} of $A$ may be maximally entangled with a bit of $B$; such an EPR (or Bell) pair is pure in the joint system and so doesn't contribute to $S(AB)$, hence contributes $2$ to $I(A:B)$, $-1$ to $H(A|B)$, and $-1$ to $H(B|A)$ (fig.\ \ref{fig:MIquantum}). Again, in the RT calculation of these quantities, it is far from clear what the difference between the areas of minimal surfaces passing through different parts of the spacetime has to do with any redundancy or cancellation between the bits in $A$ and in $B$.

A similar confusion arises for properties of entanglement entropies such as subadditivity and strong subadditivity. These fundamental properties have clear information-theoretic meanings, namely the positivity and monotonicity under inclusion of correlations, respectively. It can be proven that the RT formula obeys these properties \cite{Headrick:2007km,Headrick:2013zda}. However, the proofs, which involve cutting and gluing minimal surfaces, bear little apparent relation to the information-theoretic meanings of the properties. In the absence of such a connection, it seems almost fortuitous that the formula satisfies these properties.

Arguably, the fallacy in ascribing too much significance to the minimal surface is in thinking of its area---and therefore $S(A)$---as a local property. In fact, since $m(A)$ is defined by its minimality, its area is really a global property of the entire bulk spacetime. To emphasize this, one can write the RT formula without $m(A)$ explicitly appearing:
\begin{equation}\label{RT2}
S(A) = \frac1{4G_{\rm N}}\min_{m\sim A}\area(m)\,,
\end{equation}
where $\sim$ means homologous.

If we don't think of the bits of $A$ as ``living on'' $m(A)$, how then should we interpret the RT formula? In this note we will provide a new interpretation which clarifies the above conceptual issues. This interpretation offers a transparent relation between the EE calculated from the formula, as well as the quantities and properties derived from it, and their information-theoretic meanings. It is hoped that this interpretation will be suggestive of a new way to think about the emergence of spacetime in quantum gravity.

Let us briefly summarize the new interpretation here. We will begin by rewriting the formula in a way that does not involve the minimal surface, or indeed any surfaces at all. Instead, we will invoke the notion of a \emph{flow}, defined as a divergenceless vector field in the bulk with pointwise bounded norm; note that this is a global object, not localized anywhere in the bulk. Its flow lines can be thought of as a set of ``threads'' with a cross-sectional area of 4 Planck areas. In this picture, each thread leaving the region $A$ carries one independent bit of information about the microstate of $A$; $S(A)$ is thus the maximum possible number of threads emanating from $A$. The equivalence of this formulation to equation \eqref{RT} arises from the fact that the minimal surface acts as a bottleneck limiting the number of threads emanating from $A$; this is formalized by the so-called max-flow min-cut (MFMC) principle, a theorem originally from network theory but which we use here in its Riemannian geometry version \cite{Federer74,MR700642,MR1088184}.\footnote{The network version of MFMC was recently applied to compute EEs in a tensor-network toy model of holography \cite{Pastawski:2015qua}.} The threads thus naturally implement the holographic principle \cite{'tHooft:1993gx,Susskind:1994vu}; the entropy is computed by an area rather than a volume simply because one is counting one-dimensional rather than pointlike objects. Both entangled and classically correlated pairs of bits are naturally described in terms of these threads, along with important information-theoretic quantities like conditional entropy, mutual information, and conditional mutual information. Subadditivity and strong subadditivity follow immediately from this picture, and moreover the proofs of these properties correspond in a transparent way to their information-theoretic meanings. Unlike the minimal surfaces, the threads do not jump under continuous deformations of the region $A$.\footnote{Since, in the thread picture, the minimal surface is eliminated as a fundamental object, an interesting question is how to think about the entanglement wedge, the bulk region that interpolates between $A$ and $m(A)$ \cite{Headrick:2014cta}. In particular, recent discussions of ``subregion duality'' and ``entanglement wedge reconstruction'' have suggested that the entanglement wedge may carry the information in the reduced density matrix $\rho_A$ (see e.g.\ \cite{Czech:2012bh,Headrick:2014cta,Jafferis:2015del}). We leave consideration of this issue to future work.} The new formulation also has certain technical advantages that we will describe.

We will explain the MFMC principle and describe the new formulation of the RT formula in the next section. In section \ref{sec:threads}, we will describe the bit threads and explain how they give rise to a natural interpretation of the formula. We conclude in section \ref{sec:discussion} with a discussion of open questions. Appendix \ref{sec:mfmc} contains a mathematical review of aspects of MFMC, focusing on its Riemannian geometry version; it also gives proofs in the network setting and conjectures in the Riemannian setting of two properties of flows (continuity and nesting) that we will need.

\section{Flows}\label{sec:reformulation}

In this section, we will state the max-flow min-cut (MFMC) theorem in its Riemannian-geometry version, and then use it to give a reformulation of the Ryu-Takayanagi formula that is mathematically equivalent to \eqref{RT} but does not make reference to the minimal surface. MFMC is a standard tool in network theory, where it originated. On the other hand, the literature in the Riemannian setting is rather obscure. Therefore, in appendix \ref{sec:mfmc} we provide a short review, and discuss some relevant extensions.

\subsection{Max flow-min cut principle}\label{MFMC}

Given an oriented Riemannian manifold with boundary and a positive constant $C$, we define a ``flow'' to be a vector field $v$ satisfying the following two properties:
\begin{equation}
\nabla_\mu v^\mu=0\,,\qquad|v|\le C\,.
\end{equation}
(We do not impose any boundary condition on $v$.) We define a ``surface'' to be an oriented codimension-one submanifold, and denote the flux of $v$ through a surface $m$ by $\int_mv$:
\begin{equation}
\int_mv:=\int_m\sqrt{h}\,n_\mu v^\mu\,,
\end{equation}
where $h$ is the determinant of the induced metric on $m$ and $n$ is the unit normal vector. Let $A$ be a region\footnote{Technically, by ``region'' we mean codimension-zero submanifold.} of the boundary. The divergenceless condition implies that the flux through $A$ equals that through any homologous\footnote{If $A$ is not closed, then by ``homologous'' we mean relative to $\partial A$.} surface:
\begin{equation}
m\sim A\quad\Rightarrow\quad\int_mv=\int_Av\,.
\end{equation}
Meanwhile, the norm bound $|v|\le C$ implies $n_\mu v^\mu\le C$, so this flux is bounded by the area of $m$:
\begin{equation}\label{bound}
\int_mv\le C\int_m\sqrt{h}=C\area(m)\,.
\end{equation}
Maximizing on one side over all flows $v$ and minimizing on the other over all surfaces homologous to $A$, we therefore have\footnote{At this stage, to be mathematically correct, we should really put sup and inf instead of max and min. However, one can show under certain conditions that the supremum and infimum are achieved; see appendix \ref{sec:mfmc}.}
\begin{equation}\label{bound2}
\max_v\int_Av\le C\min_{m\sim A}\area(m)\,.
\end{equation}

\begin{figure}[tbp]
\centering
\includegraphics[width=0.5\textwidth]{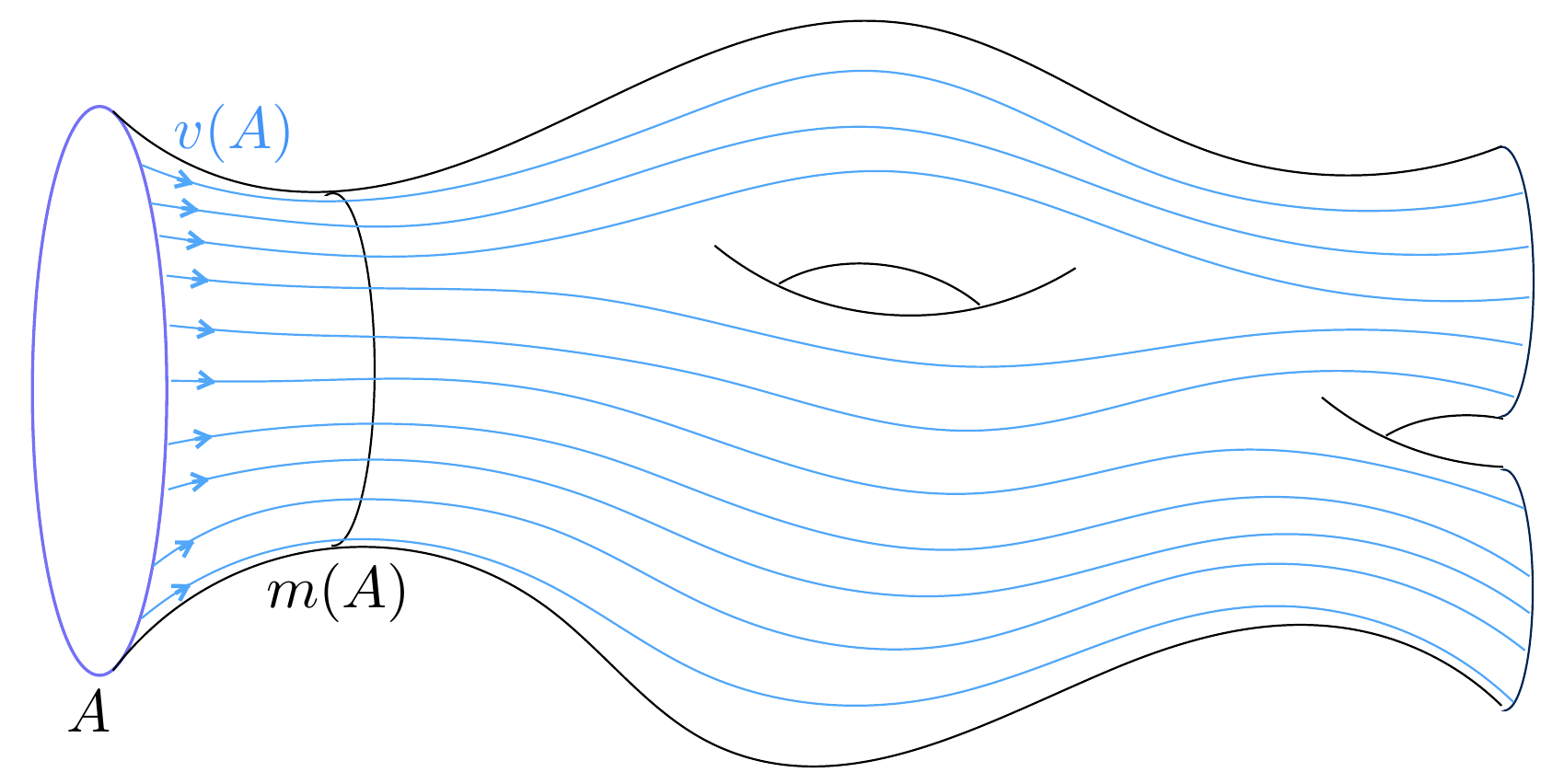}
\caption{\label{fig:maxflow}
Illustration of the Riemannian max-flow min-cut theorem: Given a boundary region $A$, the minimal-area representative of its homology class $m(A)$ is the bottleneck; its area gives an upper bound on the flux for any flow. The theorem asserts that there exists a flow $v(A)$ whose flux equals the area of $m(A)$ (times the constant $C$). In the figure, $v(A)$ is shown by its flow lines. On $m(A)$, this flow necessarily equals $C$ times the unit normal $n$.
}
\end{figure}

So far this is all fairly obvious. The max-cut min-flow (MFMC) theorem \cite{Federer74,MR700642,MR1088184} makes the non-trivial statement that the inequality \eqref{bound2} is in fact saturated:
\begin{equation}\label{mfmc}
\max_v\int_Av=C\min_{m\sim A}\area(m)\,.
\end{equation}
In other words, the inequalities \eqref{bound} for all the different members $m$ of the homology class are the only obstructions to increasing the flux; the strongest of these is obviously the area-minimizing representative $m(A)$---this is the bottleneck. Any flow that achieves the maximum flux clearly must have $n_\mu v^\mu=C$, hence $v=Cn$, everywhere on $m(A)$.\footnote{The Riemannian MFMC theorem can be phrased in the language of calibrations. Via the Hodge star,  $w=\star(v_\mu dx^\mu)$, the definition of a flow $v$ is equivalent to that of a $(d-1)$-calibration $w$ (setting $C=1$ to conform to the usual definition), and the statement ``$v=n$ on the surface $m$'' is equivalent to the statement ``$w$ calibrates $m$''. It is a standard result that if a surface is calibrated then it has minimal area in its homology class. The MFMC theorem asserts the converse: any surface that is minimal in its homology class is calibrated. While calibrated implies minimal in any codimension, the converse is special to codimension 1, as can be shown by simple counterexamples. See the appendix for further discussion.} Elsewhere, however, the constraints are weaker and there is considerable freedom in choosing the flow. Thus, whereas the area-minimizer $m(A)$ is generically unique, the flux-maximizer generically enjoys an enormous (infinite-dimensional) degeneracy.\footnote{Mathematically, for generic metrics, the max flow is \emph{underdetermined}. This explains why the topic lies outside mainstream differential geometry; there is no well-posed PDE to solve!} We will let $v(A)$ denote \emph{any} flux-maximizing flow. The theorem is illustrated in figure \ref{fig:maxflow}.

\begin{figure}[tbp]
\centering
\includegraphics[width=0.3\textwidth]{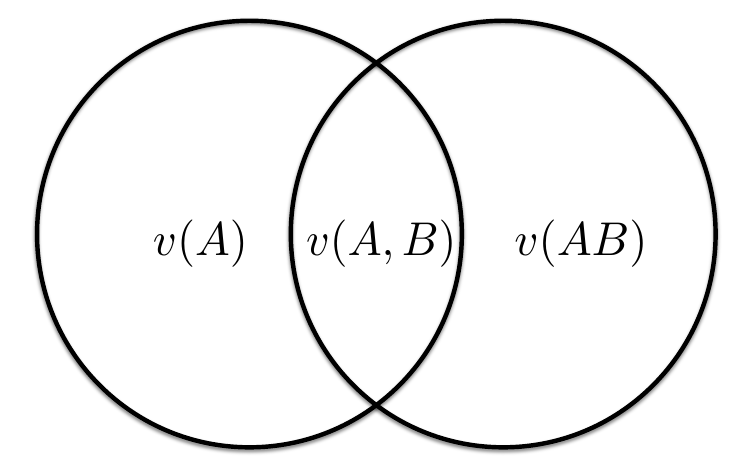}
\caption{\label{fig:Venn}
The nesting property asserts that the set $\{v(A)\}$ of flows maximizing the flux on $A$ overlaps the set $\{v(AB)\}$ of flows maximizing the flux on $AB$. Flows in the overlap are referred to as $v(A,B)$.
}
\end{figure}

Two extensions of the theorem will be useful to us in what follows. In subsections \ref{continuity} and \ref{nesting} respectively of the appendix, we will prove each property in the network setting and suggest how the proof may be carried over to the Riemannian setting.\footnote{We are not aware of proofs of these statements in the literature. However, the literature on the network version of MFMC is very extensive, and it seems likely that one or both of these properties have previously been noted in some form.} Firstly, the maximizing flow $v(A)$ varies continuously under continuous deformations of $A$; more precisely, given the degeneracy of the maximizer, it can be chosen to vary continuously. Secondly, suppose we have two regions $A$, $B$ of the boundary, which without loss of generality we assume to be disjoint. We cannot in general find a flow that maximizes the flux through both regions simultaneously. The reason is that the bottleneck $m(AB)$ for their union $AB$ may have an area smaller than the sum of the areas of $m(A)$ and $m(B)$. Then
\begin{equation}
\int_Av+\int_Bv=\int_{AB}v\le C\area(m(AB))<C\area(m(A))+C\area(m(B)),
\end{equation}
implying that either the flux through $A$ or through $B$ fails to achieve its maximum. On the other hand, there \emph{do} always exist flows that simultaneously maximize the flux through $A$ and through $AB$. We will call this the ``nesting'' property, and will denote such a flow by $v(A,B)$. Thus any flow called $v(A,B)$ could also be called $v(A)$ or $v(AB)$ (but not in general $v(B)$). (See figure \ref{fig:Venn}.) It is also useful to think of $v(A,B)$ in two other, equivalent ways: as a flow that maximizes the flux through $B$ among those that maximize the flux through $A$; and as a flow that \emph{minimizes} the flux through $B$ among those that maximize the flux through $AB$. There is an obvious generalization to more regions; for example $v(A,B,C)$ simultaneously maximizes the fluxes through $A$, $AB$, and $ABC$.

\subsection{Reformulation of Ryu-Takayanagi}\label{sec:reformulation2}

\begin{figure}[tbp]
\centering
\includegraphics[width=0.5\textwidth]{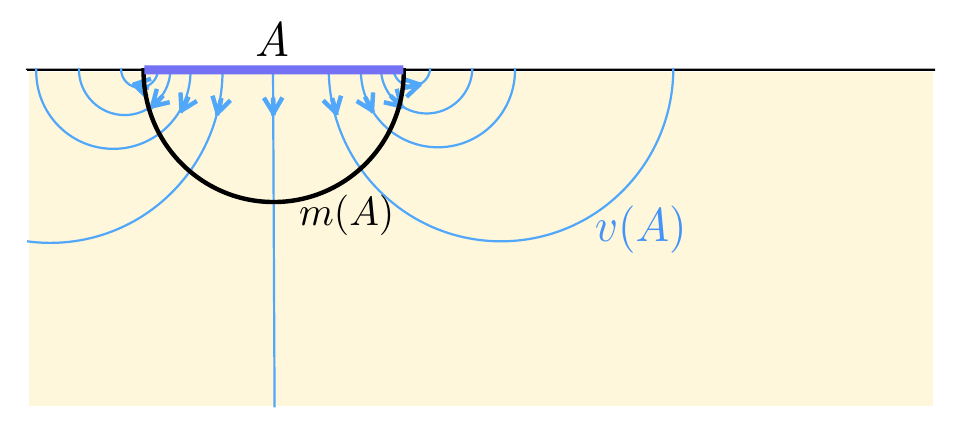}
\caption{\label{fig:maxflow2}
According to eq.\ \eqref{RTnew}, the entanglement entropy of the region $A$ is given by the maximum flux through $A$ of any flow. A maximizing flow $v(A)$ is illustrated by its flow lines in blue. This flux will equal the area of the RT minimal surface $m(A)$ (divided by $4G_{\rm N}$).
}
\end{figure}

We now return to the holographic context. The Riemannian manifold is a constant-time slice of a static bulk spacetime, $A$ is a region of its conformal boundary,\footnote{\label{boundaries}In additional to the conformal boundary where the dual field theory lives, the slice may have a boundary which is a horizon. Recall that there is no boundary condition on $v$, so it may have non-zero flux through horizons. The bulk may also end on singularities such as orbifold and orientifold fixed planes, end-of-the-world branes, and walls where internal dimensions cap off. However, as explained in \cite{Headrick:2013zda}, these do not count of ``boundaries'' for the purposes of computing holographic entanglement entropy, and therefore $v$ must have vanishing flux through them.} and we set $C=1/(4G_{\rm N})$. By \eqref{mfmc}, we can now rewrite the Ryu-Takayanagi formula \eqref{RT2} in the following simple way (see fig.\ \ref{fig:maxflow2}):
\begin{equation}\label{RTnew}
S(A) = \max_v\int_Av\,.
\end{equation}
It is worth noting that both the global minimization and the homology condition in the usual formulation are automatically incorporated in \eqref{RTnew}. Returning to the question in the introduction, ``How should we think about the minimal surface?'', the answer is just that it serves as the bottleneck for the flow. If the region $A$ is varied, the bottleneck can jump even while the flow changes continuously.

If the region has a non-empty entangling surface $\partial A$, then the entanglement entropy (EE) will have an ultraviolet divergence. (There may also be an infrared divergence, if there is a finite entropy density and the region is infinite in extent. Similar remarks to those below apply to that case.) In the minimal-surface picture, this is due to the divergent surface area near the boundary; in the flow picture, it is due to the divergent flow near the entangling surface. It is then necessary to regulate the divergence by moving the boundary to a finite value of the radial coordinate. There is, however, an interesting difference between the surface and flow pictures in this regard. In the surface picture, it is necessary to introduce a regulator even to \emph{define} the minimal surface. In particular, while one can define a locally minimal surface even if its area is infinite (namely a surface whose local area increases under local variations), one cannot determine which of several such surfaces should be considered the global minimum. On the other hand, we can give a definition of a maximal flow that applies even if the flux is infinite. First, given a flow $v$ we define an ``augmentation'' as a vector field $\Delta v$ such that $v+\Delta v$ is also a flow and $\Delta v$ has positive flux through $A$. A maximal flow is then one that does not admit an augmentation. The utility of this definition will become clear when we discuss the mutual information in subsection \ref{sec:tworegions} below.

The conceptual implications of \eqref{RTnew} are the main focus of this paper. However, as an aside we note that this formula may actually be useful for the numerical evaluation of holographic EEs. Finding a max flow requires maximizing a linear functional on a vector space (the space of divergenceless vector fields) subject to a convex constraint; in other words, it is a convex optimization problem. This is in contrast to the problem of finding the min cut, which requires finding the global minimum of a functional that is defined on a non-linear space and typically has local minima.\footnote{Actually, min cut can also be turned into a convex problem, as follows. (This is an example of ``convex relaxation''.) One considers a real function $\psi$ on the manifold, subject to the constraint $\psi|_A=1$, $\psi|_{A^c}=0$, and minimizes the functional $\mathcal{F}[\psi]:=\int\sqrt{g}|\nabla\psi|$. 
On the minimum, $\psi = 1$ on the ``entanglement wedge'' $r(A)$ (the bulk region that interpolates between $A$ and $m(A)$) and 0 on its complement; hence $|\nabla\psi|$ is a delta function supported on $m(A)$, and $\mathcal{F} = \area(m(A))$. Min cut, in this form, is related to max flow by the so-called ``strong duality'' of convex optimization problems. A fuller explanation of this will be given in \cite{covariantflows}.} For that reason, for certain classes of computational problems such as image processing, flow maximization is often used as a method for finding minimal surfaces. The basic strategy is the so-called Ford-Fulkerson algorithm: Start with an arbitrary flow and augment it until it can't be augmented anymore. We leave the investigation of possible numerical applications of \eqref{RTnew} to future work.

Given multiple subsystems of a quantum system, there are several linear combinations of EEs that have information-theoretic significance, and can easily be evaluated using the RT formula. The most important of these are the conditional entropy, mutual information, and conditional mutual information. The crucial subadditivity and strong subadditivity (SSA) inequalities are most simply expressed in terms of these quantities. In the rest of this section, we will see that these quantities and properties are naturally expressed in terms of flows.

\subsection{Two regions}\label{sec:tworegions}

We begin with the conditional entropy $H(A|B):=S(AB)-S(B)$, which has a simple and useful expression in terms of flows. Using the nesting property explained in subsection \eqref{MFMC}, we choose a flow $v(B,A)$ that simultaneously maximizes the flux through $B$ and through $AB$. Then
\begin{equation}
H(A|B)=\int_{AB}v(B,A)-\int_Bv(B,A)=\int_Av(B,A)\,.
\end{equation}
Thus it is the \emph{minimum} possible flux through $A$, the amount left over after as much as flux possible has been put on $B$, given that the $AB$ flux has been maximized. (Note that this amount may be negative; we will discuss examples in the next section.)

The mutual information $I(A:B):=S(A)+S(B)-S(AB)=S(A)-H(A|B)$ is the difference between the maximum and minimum flux on $A$,
\begin{equation}
I(A:B)=\int_Av(A,B)-\int_Av(B,A)\,,
\end{equation}
i.e.\ the amount of flux that can be shifted from $A$ over to $B$ (again, always maximizing on $AB$). If it is zero, then the flow that maximizes the flux on $AB$ and $B$ also maximizes on $A$; this implies that the $AB$ bottleneck simply consists of the union of the $A$ and $B$ bottlenecks.

The fact that $I(A:B)$ is the difference between the maximum and minimum fluxes through $A$ (subject to maximizing on $AB$) immediately implies that it is non-negative; this is subadditivity of EE. It can also be proved without appealing to the nesting property.\footnote{We thank V. Hubeny for helpful discussions on this and related points.} Simply pick any flow $v(AB)$ that maximizes on $AB$; by \eqref{RTnew}, its flux through $A$ cannot exceed $S(A)$, and similarly for $B$, so we have
\begin{equation}\label{SAnonesting}
S(A)+S(B)\ge\int_Av(AB)+\int_Bv(AB) = \int_{AB}v(AB)=S(AB)\,.
\end{equation}

If the regions $A,B$ do not share a boundary, then the ultraviolet divergences in their EEs are additive in $S(AB)$, so the mutual information is ultraviolet-finite. To calculate or even define this quantity using the minimal-surface formulation of the RT formula requires first introducing and then removing a regulator. However, as discussed in the previous subsection, in the flow picture we can define maximal flows even when they have an infinite flux. Using this, we can define the mutual information directly in the unregulated theory. As above, we let $v(A,B)$ be a maximal flow through $A$ and $AB$ (i.e.\ one for which there exists neither an augmentation $\Delta v$ such that $\int_A\Delta v>0$ nor one such that $\int_{AB}\Delta v>0$), and similarly for $v(B,A)$. We then define
\begin{equation}\label{MIflow}
I(A:B)=\int_A\left(v(A,B)-v(B,A)\right).
\end{equation}
This will agree with the definition from introducing and removing a regulator. Thus, even if the total fluxes are infinite, the amount of flux that can be shifted between $A$ and $B$ is a well-defined quantity. 

\begin{figure}[tbp]
\centering
\includegraphics[width=0.5\textwidth]{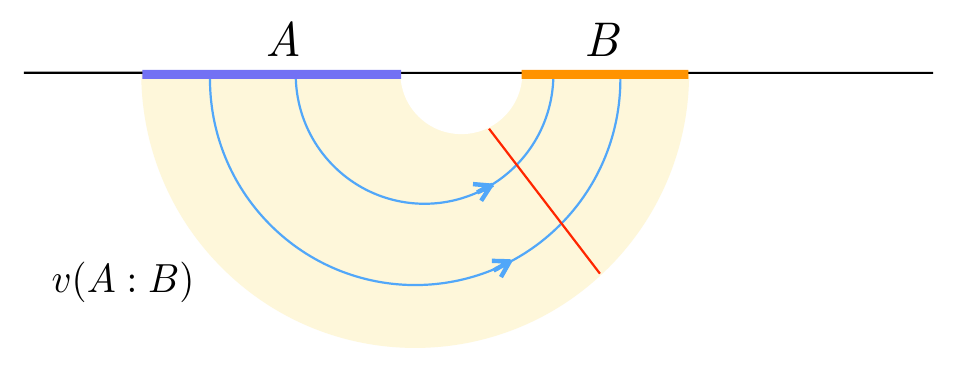}
\caption{\label{fig:MIflow}
The flow $v(A:B)$ defined by equation \eqref{MIflow} vanishes outside the entanglement wedge $r(AB)$ (shaded region), which is a tube connecting $A$ and $B$. Its flux, which is half the mutual information, is bounded above by the area of the neck of this tube (shown in red). 
}
\end{figure}

Equation \eqref{MIflow} leads to an interesting connection between the mutual information and the entanglement wedge $r(AB)$, the bulk region that interpolates between $AB$ and $m(AB)$. 
First, note that the vector field
\begin{equation}\label{MIflow}
v(A:B):=\frac12\left(v(A,B)-v(B,A)\right),
\end{equation}
is itself a flow. Since $v(A,B)$ and $v(B,A)$ both equal the unit normal on $m(AB)$, $v(A:B)$ vanishes there. Furthermore, outside of $r(AB)$, $v(A,B)$ and $v(B,A)$ can be chosen equal (since they are subject to the same constraints), making $v(A:B)$ vanish. Thus we can assume that $v(A:B)$ is non-zero only inside $r(AB)$, which is a tube connecting $A$ and $B$. Since $v(A:B)$ is a flow, its flux---which is half the mutual information---is bounded above by the area of the ``neck'' of this tube, the least-area surface in $r(AB)$ separating $A$ and $B$ (technically, the least-area surface in $r(AB)$ homologous to $A$ relative to $m(AB)$).\footnote{This statement can also be proven using minimal surfaces rather than flows by appropriately cutting up $m(AB)$.} This situation is illustrated in fig.\ \ref{fig:MIflow}. As an example, for two intervals $A=[-y,-1]$, $B=[1,y]$ on the boundary of AdS${}_3$, $r(AB)$ is connected as long as $y>3+2\sqrt2$. The area of the neck is $\frac c6\ln y$, which is larger than half the mutual information, $\frac12I(A:B) = \frac c3\ln((y-1)/(2\sqrt y))$. (In the regime where $r(AB)$ is disconnected, both the neck and the mutual information vanish, so the bound is trivial.)

The constraints $\nabla_\mu v^\mu=0$, $|v|\le1/(4G_{\rm N})$ that define a flow are invariant under $v\to-v$. Therefore, in addition to the definitional upper bound $\int_Av\le S(A)$, we have the lower bound $\int_Av\ge-S(A)$. In other words, the bottleneck constrains the flux in either direction. We can use this to give a proof of the Araki-Lieb inequality $|S(B)-S(A)|\le S(AB)$, similar to the proof \eqref{SAnonesting} for subadditivity. Let $v(B)$ be any flow maximizing the flux out of $B$. Then
\begin{equation}
S(AB)+S(A) \ge \int_{AB}v(B)-\int_Av(B) = \int_Bv(B)=S(B)\,.
\end{equation}

\subsection{Three regions}\label{sec:threeregions}

The most important quantity involving three subsystems is the conditional mutual information $I(A:B|C)$, so-called because (classically) it is the mutual information between $A$ and $B$, conditioned on $C$:
\begin{equation}
I(A:B|C):=S(AC)+S(BC)-S(C)-S(ABC)\,.
\end{equation}
This can be written in various ways in terms of conditional entropies or mutual informations. For example, we can write it as
\begin{equation}
I(A:B|C) = H(A|C)-H(A|BC)
 = \int_Av(C,A) - \int_Av(BC,A)\,.
\end{equation}
Invoking the nesting property again, we can assume without loss of generality that the first flow also maximizes on $ABC$ and the second also maximizes on $C$; then we have
\begin{equation}\label{CMI}
I(A:B|C) = \int_Av(C,A,B) - \int_Av(C,B,A)\,.
\end{equation}
The first term is the maximum possible flux through $A$ and the second term the minimum, always subject to the constraint of first maximizing on $C$ and $ABC$. Thus the conditional mutual information is the amount of flux that can be shifted between $A$ and $B$ subject to those constraints.

Since the maximum cannot be less than the minimum, we have $I(A:B|C)\ge0$, which is SSA. This proof that RT obeys SSA is in some ways simpler than the one based on cutting and pasting minimal surfaces, for which various special cases must be taken into account \cite{Headrick:2007km,Headrick:2013zda}. More importantly, as we will discuss in the next section, the present proof relates in a transparent manner to the information-theoretic meaning of SSA.

\begin{figure}[tbp]
\centering
\includegraphics[width=0.5\textwidth]{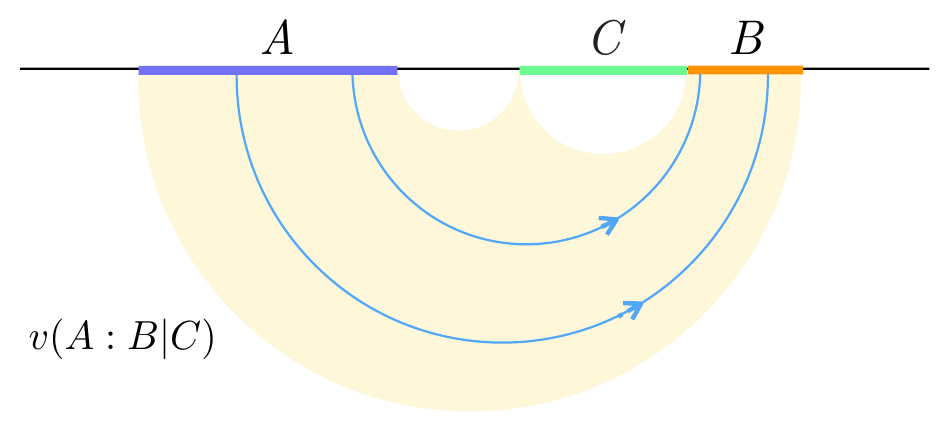}
\caption{\label{fig:CMIflow}
The flow $v(A:B|C)$ defined by equation \eqref{CMIflow} vanishes outside the region $r(ABC)\setminus r(C)$ (shaded region), which is a tube connection $A$ and $B$. Its flux, which is half the conditional mutual information, is bounded above by the smallest area separating $A$ and $B$ in this tube.
}
\end{figure}

As with the mutual information, the conditional mutual information can be finite even when its component EEs are divergent; by rewriting \eqref{CMI} as
\begin{equation}
I(A:B|C) = \int_A\left(v(C,A,B) - v(C,B,A)\right),
\end{equation}
the flow again provides a regulator-free definition. Furthermore, the flow
\begin{equation}\label{CMIflow}
v(A:B|C):=\frac12\left(v(C,A,B)-v(C,B,A)\right)
\end{equation}
can be chosen to vanish outside of the region $r(ABC)\setminus r(C)$, which is a tube connecting $A$ and $B$.\footnote{We thank N. Bao for pointing this out.} Its flux, which is half the conditional mutual information, is bounded above by the area of the neck of this tube (the least-area surface separating $A$ and $B$); see figure \ref{fig:CMIflow}. Note that
\begin{equation}
v(A:B|C) = v(A:BC)-v(A:C)\,,
\end{equation}
(for appropriate choices of $v(A:BC)$ and $v(A:C)$), which is the analogue at the level of flows of the relation $I(A: B|C)=I(A:BC)-I(A:C)$ for fluxes.

The last linear combination of entropies we will consider is the tripartite information:
\begin{equation}
I_3(A:B:C):=S(A)+S(B)+S(C)+S(ABC)-S(AB)-S(AC)-S(BC)\,.
\end{equation}
Like the conditional mutual information, this can be rewritten in various ways in terms of mutual informations or conditional entropies---for example $I_3(A:B:C) = I(A:B)+I(A:C)-I(A:BC)$---and therefore in terms of flows. We leave this as an exercise for the reader. Our main interest here is in the fact that, in holographic theories, the tripartite information is always non-positive,
\begin{equation}\label{MMI}
I_3(A:B:C)\le0
\end{equation}
a property called\footnote{The name ``monogamy of mutual information'' is perhaps slightly obscure, as the connection between this inequality and monogamy of entanglement is rather indirect. A clearer name might be ``superadditivity of mutual information''.}  ``monogamy of mutual information'' (MMI) \cite{Hayden:2011ag,Headrick:2013zda}. The proof is similar to the original one for SSA, involving cutting and pasting minimal surfaces (again, various special cases have to be considered).\footnote{An infinite set of inequalities involving more than three regions has recently been discovered, generalizing MMI \cite{Bao:2015bfa}.} We have not been able to find a proof of this property in the flow language (i.e., that does not invoke MFMC to pass back to the minimal surface and then apply the known proof there). We will discuss this further below.

\begin{figure}[tbp]
\centering
\includegraphics[width=0.7\textwidth]{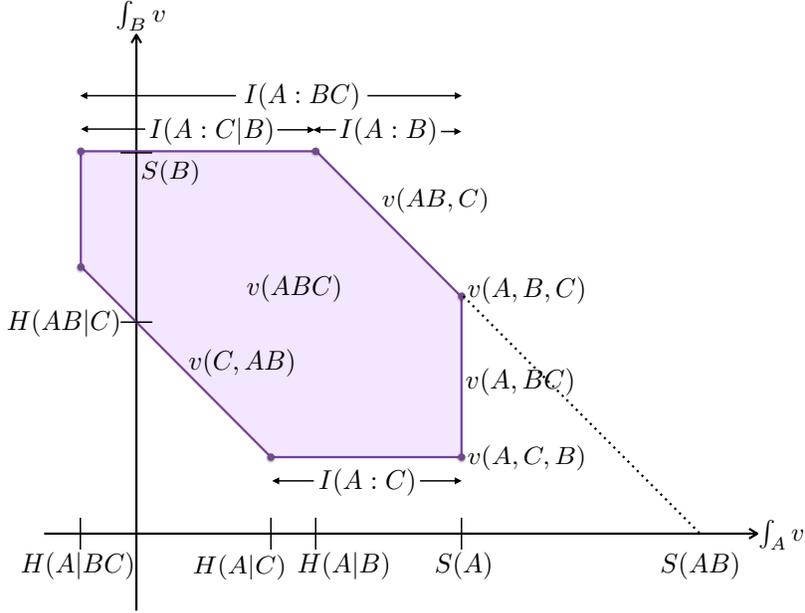}
\caption{\label{fig:hexagon}
Given a bulk geometry and three boundary regions $A,B,C$, the set of flows $v(ABC)$ that maximize the flux through $ABC$ can be plotted on a plane with the fluxes through $A$ and $B$ respectively as coordinates. The result is in general a hexagon or lower polygon where each side is either vertical, horizontal, or at a $45^\circ$ angle. The figure shows an illustrative example. Points in the interior of the hexagon represent flows that maximize only the flux through $ABC$; points lying on an edge represent flows that maximize also one other flux (e.g.\ $v(AB,C)$, $v(A,BC)$, $v(C,AB)$, shown on the diagram); and the vertices represent flows that maximize two other fluxes (e.g.\ $v(A,B,C)$, $v(A,C,B)$ shown). As shown in the figure, most quantities of interest that can be derived from the EEs of these regions---conditional entropies, mutual informations, and conditional mutual informations---are represented as positions or distances on this diagram. The tripartite information $I_3(A:B:C)=I(A:C)-I(A:C|B)$ is the difference in length between the bottom and of the top of the hexagon. In the example shown, this is negative as required by the MMI inequality. However, this is not the case for all hexagons (a simple counterexample is given in the main text).
}
\end{figure}

A useful way to visualize the various EEs and their linear combinations is shown in figure \ref{fig:hexagon}. We consider the set of all possible flows $v(ABC)$ that maximize the flux on $ABC$, and plot them on a plane using as coordinates their fluxes on $A$ and $B$ respectively. Note that, given the total flux through $ABC$ (which is $S(ABC)$), all other fluxes (through $AB$, $C$, etc.) are determined by those two. The possible fluxes will fill out a hexagon (or lower polygon) with edges that are horizontal, vertical, or at a $45^\circ$ angle running northwest-southeast.\footnote{Like the toric diagram of a del Pezzo surface.} As shown in the figure, many quantities of interest are represented by the positions of vertices and lengths of edges on this hexagon. Subadditivity and SSA are clear from the fact that these lengths cannot be negative. 

One can ask conversely whether a given hexagon can be realized as a set of fluxes for some actual geometry. Positivity of the entropies and the Araki-Lieb inequality impose some constraints on the positions of the edges. These, however, are not enough to enforce the MMI inequality \eqref{MMI}. A simple counterexample is a triangle with vertices at $(0,0)$, $(1,0)$, $(0,1)$; setting $S(ABC)=1$, this triangle represents the following entropies:
\begin{equation}
S(A)=S(B)=S(C)=S(AB)=S(AC)=S(BC)=1\,.
\end{equation}
These entropies are realized by the state
\begin{equation}
\rho_{ABC} = \frac12\left(\ket{000}\bra{000}+\ket{111}\bra{111}\right)
\end{equation}
(which can be purified to the 4-party GHZ state $(\ket{0000}+\ket{1111})/\sqrt{2}$). However, since $I_3(A:B:C)=1$, \eqref{MMI} is violated, and so the corresponding fluxes cannot be realized by any geometry. Since nesting and other basic properties of fluxes are implicitly satisfied in the hexagon construction, this counterexample shows that those properties are not sufficient to prove MMI. It follows that flows obey some other non-trivial property beyond nesting, which would be very interesting to discover.

\section{Interpretation}\label{sec:threads}

Our purpose in this section is to attach an interpretation---essentially, a set of pictures---to the right-hand side of equation \eqref{RTnew}, which will connect it to the information-theoretic meaning of its left-hand side, the entanglement entropy, as well as derived quantities and concepts like mutual information, subadditivity, etc. The interpretation, in terms of so-called bit threads, is explained in subsection \ref{sec:bitthreads} and expanded upon in subsection \ref{sec:bitthreads2}. These are followed in subsection \ref{sec:Weyl} by a further, more speculative interpretation, which relates bit threads to Weyl's law from harmonic analysis.

\subsection{Bit threads}\label{sec:bitthreads}

As with an electric, magnetic, or fluid velocity field, it is convenient to visualize the flow $v$ by its field lines. These are defined as a set of integral curves of $v$ chosen so that their transverse density equals $|v|$. We will call these flow lines ``bit threads'', for a reason that will become clear soon.\footnote{``Qubit threads'' would perhaps be more accurate, but seems awkward.} Please keep in mind that the threads are oriented.

The bit threads inherit two important properties from the definition of a flow. First, the bound $|v|\le1/(4G_{\rm N})$ means that they cannot be packed together more tightly than one per 4 Planck areas. Thus they have a microscopic but nonetheless finite thickness. In general, their density on macroscopic (i.e.\ AdS) scales will be of order $N^2$ (in the usual gauge/gravity terminology). Therefore, unless we are interested in $1/N$ effects (which we will mostly ignore in this paper), we should not worry too much about the discrepancy between the continuous flow and the discrete threads. Second, the condition $\nabla_\mu v^\mu=0$ means that the threads cannot begin, end, split, or join in the bulk; each thread can begin and end only on a boundary, which could be the conformal boundary where the field theory lives, or possibly a horizon (e.g.\ if we are considering a single-sided black hole spacetime).\footnote{However, the threads cannot end on singularities in the bulk; see the discussion in footnote \ref{boundaries}.}

We would now like to suggest that a thread that emanates from a region $A$ on the boundary (and does not return to it) should be thought of as a channel that can carry one independent bit\footnote{We remind the reader that, as explained in footnote \ref{bits}, by ``bit'' we mean ``$\ln 2$ bits''.} of information about the microstate of $A$. The maximum number of independent bits is the entropy $S(A)$.  This gives an interpretation to \eqref{RTnew}. The rest of this section will be devoted to developing this interpretation further, and using it to resolve the conceptual puzzles surrounding the RT formula that were described in the introduction.
 
 We begin with a few general comments. First, as emphasized in the previous section, the maximum allowed number of threads leaving $A$ is a global property of the bulk spacetime. The minimal surface $m(A)$ is the place where they happen to be most tightly packed together---where the bits are literally compressed, to their maximum allowed of 1 per 4 Planck areas. Under continuous deformations of $A$, even when the location of this bottleneck jumps, the thread configuration change in a continuous manner. This resolves one of the conceptual puzzles.

Also as emphasized in the previous section, even given the constraint of maximizing the number of threads leaving $A$, there remains considerable freedom in choosing the configuration, and in particular where to attach the threads in $A$. Given the large volume near the boundary, there is also considerable freedom to add extra threads that begin and end on $A$ without changing the net number leaving $A$. (Note, however, that these extra threads cannot cross the minimal surface, as there is no room for them there.) The freedom to move threads around and to add extra ones is a kind of gauge freedom, which, as we will explain in the next subsection, has an important physical significance.

Third, it is well known that a region $A$ with a non-empty entangling surface $\partial A$ will have a divergent entropy $S(A)$. In the usual formulation of RT, this is due to the infinite area of the minimal surface near the boundary. In this picture, it is due to the fact that an infinite number of threads can be squeezed into the bulk near the entangling surface.

Finally, as a simple example, consider a one-sided black hole, representing a mixed state of the field theory. The entropy of the entire boundary is simply the black hole's entropy, which (by either the Bekenstein-Hawking or the RT formula, since the horizon is the minimal surface) is the area of the horizon. In the bit-thread picture, the only threads that count are those that leave the boundary and don't return; since they can't end in the bulk, such threads must end on the horizon, which is itself the bottleneck. If we consider the two-sided version of the same black hole, but still evaluate the entropy of one boundary, then the threads continue to the other side and end on the other boundary, with the bottleneck still being the horizon.

\subsection{Correlations and entanglement}\label{sec:bitthreads2}

We now consider two disjoint regions $A$, $B$. We start with the case where the joint system $AB$ is pure. Then the amount of entanglement is equivalent to $S(A)$ ($=S(B)$) EPR pairs. Since the total flux on $AB$ must vanish, any thread leaving $A$ must either return to it or end on $B$. The maximum that may go from $A$ to $B$ is the maximum flux on $A$, which is $S(A)$. (In general one may consider thread configurations with threads going both ways. The flux measures the net number: the number going from $A$ to $B$ minus the number going the other way. However,  when the maximum flux is achieved, no threads may go the other way, as the threads going from $A$ to $B$ already occupy the entire bottleneck $m(A)$.) We can simply reverse the direction of all of these threads in order to obtain a configuration maximizing the number leaving $B$. Thus an entangled pair of bits is represented by a thread connecting $A$ and $B$, which switches direction depending on which entropy is being measured.

\begin{figure}[tbp]
\centering
\includegraphics[width=\textwidth]{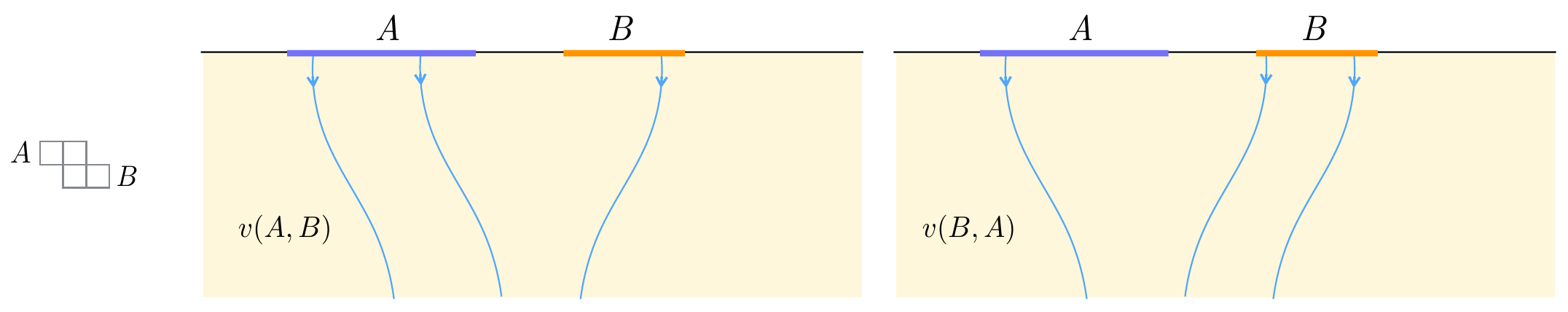}
\caption{\label{fig:toy1}
Allowed thread configurations for the toy example discussed in the text, with $S(A)=S(B)=2$, $S(AB)=3$. One thread is stuck on $A$, representing a bit that (in the compressed encoding of $AB$) is unique to $A$, one is stuck on $B$, representing a bit that is unique to $B$, and one can be moved between them, representing a pair of classically correlated bits. The box diagram on the left represents this state in the scheme of figure \ref{fig:MIclassical}.
}
\end{figure}

\begin{figure}[tbp]
\centering
\includegraphics[width=\textwidth]{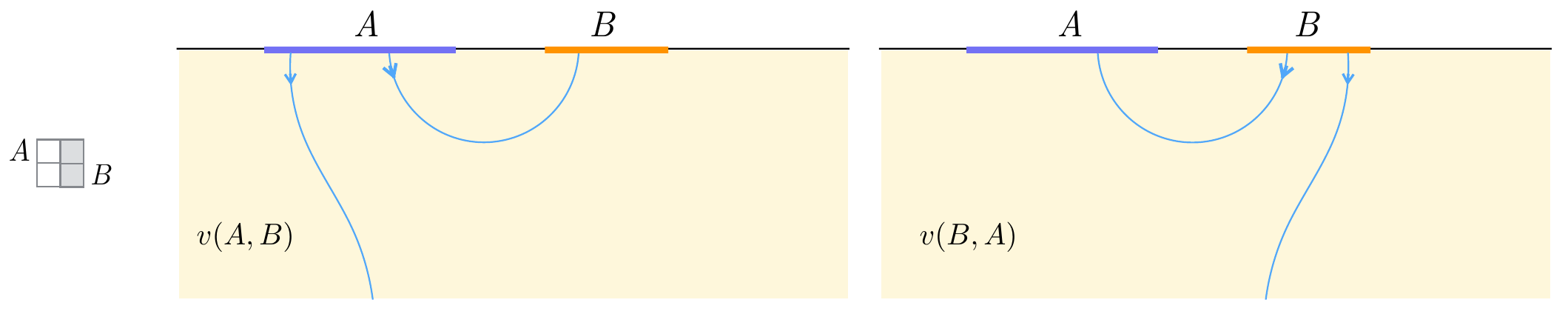}
\caption{\label{fig:toy3}
Allowed thread configurations for a toy example with $S(A)=S(B)=2$, $S(AB)=1$. In the configuration $v(A,B)$ (left), although two threads can leave $A$, only one can leave $AB$, so the second must go to $B$. In $v(B,A)$ (right), the first thread moves over to $B$, while the second switches direction. The moving thread represents a classically correlated pair of bits, while the connecting one represents an entangled pair. The box diagram on the left represents this state in the scheme of figure \ref{fig:MIquantum}.
}
\end{figure}

Classical correlations between $A$ and $B$ require $S(AB)$ to be nonzero. In such a case some threads---up to $S(AB)$ total---can leave $A$ or $B$ and end elsewhere (neither on $A$ nor $B$). (Again, when that number is saturated, the bottleneck $m(AB)$ will be fully occupied, so no threads can come the other way, beginning elsewhere and ending on $A$ or $B$.) Consider first a toy example, in which $S(A)=S(B)=2$ and $S(AB)=3$. We then have $H(A|B)=H(B|A)=I(A:B)=1$, so one bit of $AB$ (in its compressed form) is unique to $A$, one is unique to $B$, and one is redundant. How is this reflected in the thread configurations? With three threads total leaving $AB$, either two can come from $A$ and one from $B$ or vice versa (fig.\ \ref{fig:toy1}). Thus one thread is stuck to $A$, representing the bit unique to $A$; one is stuck to $B$, representing the bit unique to $B$; and one is free to move between $A$ and $B$, representing the redundant bit. In general, as long as both conditional entropies $H(A|B)$ and $H(B|A)$ are non-negative, in the thread configurations with the maximum $S(AB)$ threads leaving $AB$, $H(A|B)$ of them will be stuck to $A$, $H(B|A)$ will be stuck to $B$, and $I(A:B)$ will be free to move between them. On the other hand, if a conditional entropy is negative, for example $H(A|B)<0$, then in $v(A,B)$ some threads leaving $A$ must end on $B$ (since $S(AB)<S(A)$); this reflects the fact that a negative conditional entropy implies the presence of entanglement. See fig.\ \ref{fig:toy3} for an example.

Let us recap what we have learned so far: We consider the set of thread configurations that maximize the total number leaving $AB$. An entangled pair of bits is represented by a thread that connects $A$ to $B$ and can switch direction. A classically correlated pair is represented by a thread that leaves $AB$ and can begin on either $A$ or $B$. And a bit that is unique to $A$ ($B$) is represented by a thread that leaves $AB$ and is stuck to $A$ ($B$).

In the previous subsection, we promised to explain the significance of two ``gauge'' freedoms that occur when evaluating $S(A)$ using threads: the freedom to choose where to attach threads to the boundary, and the freedom to add extra threads that begin and end on $A$.  If we divide $A$ arbitrarily into two subregions, $A=A_1A_2$ and apply the lessons of the previous paragraph, we learn that the freedom to move the threads around reflects the existence of classical correlations between different spatial locations in the field theory, while the freedom to add extra threads reflects the existence of entanglement between different locations.

\begin{figure}[tbp]
\centering
\includegraphics[width=\textwidth]{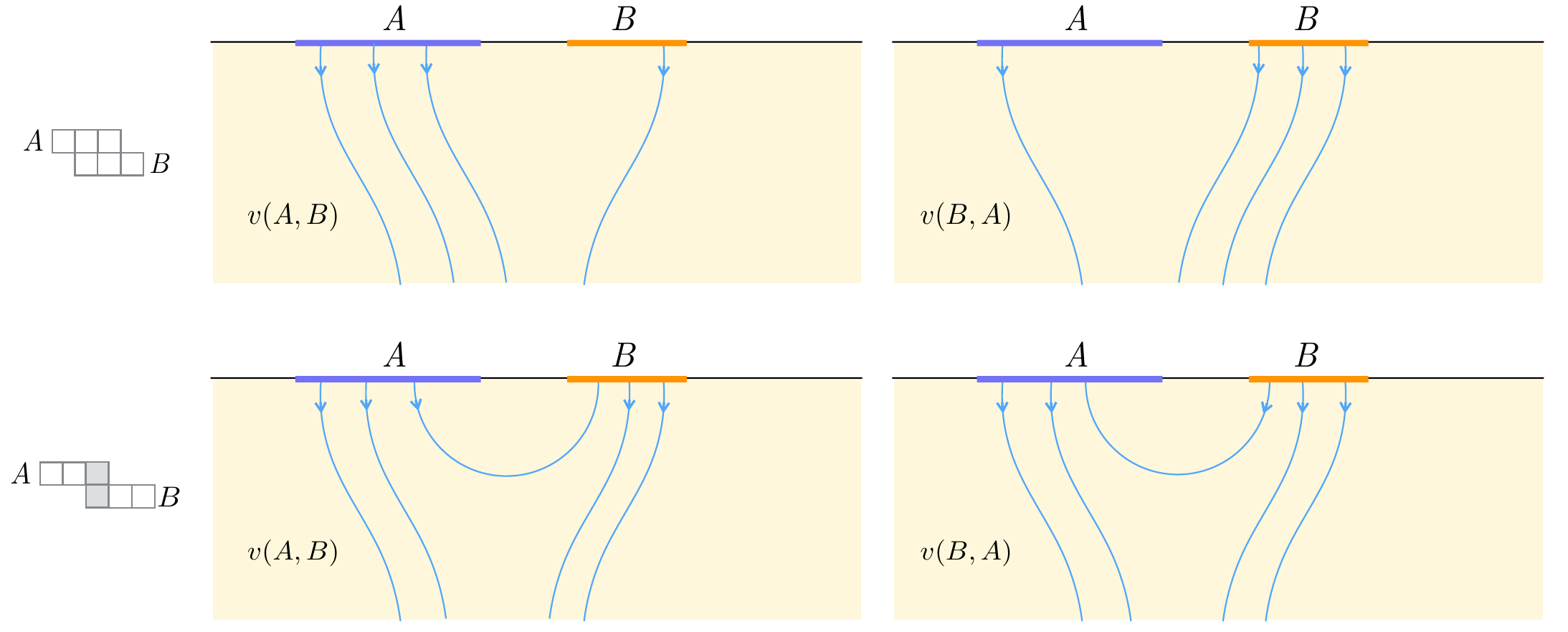}
\caption{\label{fig:toy2}
Two possible ways of arranging bits given the entropies $S(A)=S(B)=3$, $S(AB)=4$, as discussed in the second toy example of the text. Top row: one bit is unique to each of $A$ and $B$ and two pairs of bits are classically correlated; in the corresponding set of thread configurations, one thread is stuck to each of $A$ and $B$ while two are free to move between them. Bottow row: two bits are unique to each and one pair of bits is entangled; in the corresponding set of thread configurations, two threads are stuck to each of $A$ and $B$ and one connects them and switches direction. In the box diagrams on the left, as in figs.\ \ref{fig:MIclassical} and \ref{fig:MIquantum}, unshaded stacked boxes represent correlated bits and shaded ones entangled bits.
}
\end{figure}

Now, as emphasized in the introduction, the thread picture is just a rewriting of the RT formula, which tells us the entropy of any given region. These entropies alone, however, cannot always distinguish between entanglement and classical correlation. For example, while a negative conditional entropy implies the presence of entanglement, a positive one does not imply its absence. Let us see, in another toy example, how the threads can accommodate either possibility. We take $S(A)=S(B)=3$ and $S(AB)=4$, so $H(A|B)=H(B|A)=1$ and $I(A:B)=2$. There are two bit configurations that would both account for these entropies: one bit unique to each of $A$ and $B$ plus two correlated pairs; or two bits unique to each and one entangled pair. Indeed, the threads allow both options: either one thread attached to each of $A$ and $B$ and two that are free to move between them; or two threads attached to each and one connecting them (fig.\ \ref{fig:toy2}).

Finally, we consider three regions. The strong subadditivity property $I(A:BC)\ge I(A:C)$ says that $A$ has at least as much total correlation (including classical correlation and entanglement) with $BC$ as with $C$, i.e.\ that the amount of correlation is monotonic under inclusion. In the thread picture, this is represented by the intuitive fact, proven in subsection \ref{sec:threeregions}, that at least as many threads can be moved or connected between $A$ and $BC$ as between $A$ and $C$.

\begin{figure}[tbp]
\centering
\includegraphics[width=0.5\textwidth]{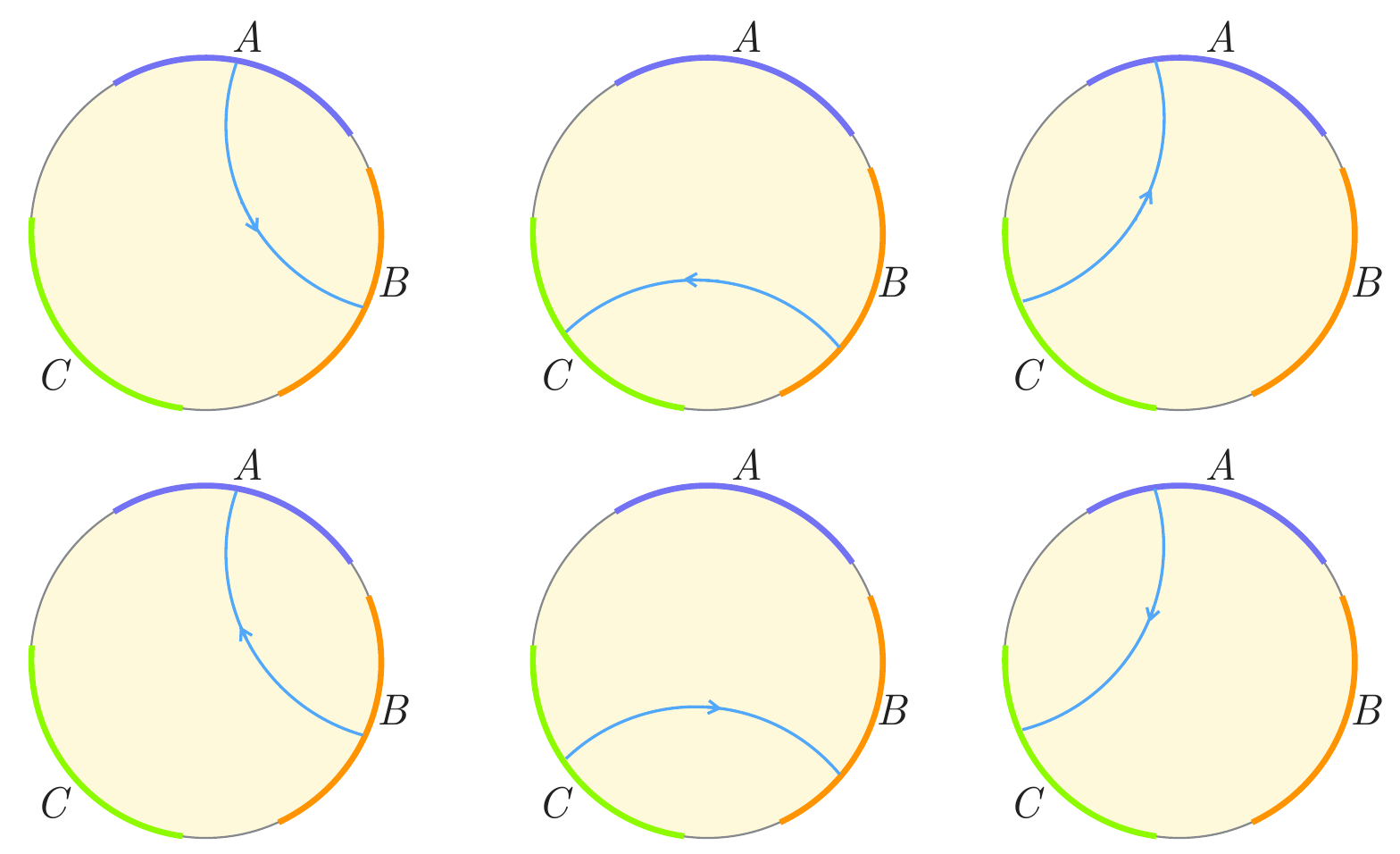}
\caption{\label{fig:GHZ}
Bit-thread representation of the GHZ state on three bits \eqref{GHZ}. Given the entropies of the different subsets of $ABC$, there are six allowed thread configurations. (A seventh, with no threads, is also allowed, of course.)
}
\end{figure}

We would like to make one final point concerning the bit-thread interpretation. Naively, given that each thread represents an entangled pair of bits, it would seem that this picture privileges bipartite entanglement, and perhaps even suggests that more complicated forms of entanglement do not occur holographically. To see that this is not the case, consider a final toy example, the GHZ state on three bits:
\begin{equation}\label{GHZ}
\ket{\text{GHZ}} = \frac1{\sqrt2}\left(\ket{000}+\ket{111}\right)\,.
\end{equation}
This state does not violate any of the known constraints on holographic EEs, and indeed it is has been argued that entanglement of this type occurs in multiboundary black holes \cite{Balasubramanian:2014hda}. The EEs for this state can be reproduced by a collection of six bit-thread configurations as shown in fig.\ \ref{fig:GHZ}. Thus, the threads certainly do accommodate multipartite entanglement. This example illustrates two important points: First, a given state is not represented by a single thread configuration, but rather a collection of configurations. Second, a thread connecting two regions only represents an entangled pair when the configuration maximizes the total number leaving the union of the regions. Thus, for example, in the top left and bottom left panels of fig. \ref{fig:GHZ}, a thread connects $A$ and $B$. However, it is not true that in the GHZ state $A$ and $B$ are entangled (tracing over $C$ leaves $A$ and $B$ merely classically correlated); indeed, those configurations do not maximize the number of threads leaving $AB$. Rather, that thread represents the entanglement between $A$ and $BC$ (as well as between $B$ and $AC$).

On the other hand, the monogamy of mutual information inequality \cite{Hayden:2011ag,Headrick:2013zda} and its generalizations \cite{Bao:2015bfa} show that not all patterns of multipartite entanglement are allowed holographically. Unfortunately, the meaning of these constraints remains obscure, and the proofs based on minimal surfaces are not very helpful. It seems reasonable to hope that, if proofs based on flows or threads can be found, they might clarify this meaning.

\subsection{Connection to Weyl's law}\label{sec:Weyl}

Having persuaded the reader that an incompressible flow from $A$ to $A^c$ is a viable alternative to the usual minimum area formulation of the  Ryu-Takayanagi idea, we would like to step back view our own suggestion a bit critically. The suggestion is certainly harmless as max-flow min-cut defines an equivalence between two dual pictures (more on this duality in appendix \ref{sec:mfmc}).  But the dual language appears to admit natural quantum mechanical corrections, so perhaps it can be formulated to make a deeper connection with quantum information. Can we see the bit-thread picture admitting the  flexibility to treat the generalization to higher-derivative gravities (e.g.\ Gauss-Bonnet) where minimum area is replaced by the minimum of another local functional $\int dA( 1+\alpha R)$.  In this section we will discuss tantalizing hints that the answer to these questions is ``Yes''.

We have already suggested a discrete interpretation for the incompressible flow $v$ as an enormous but finite collection of lines representing qubit entanglement---in the simplest manifestation each line can be thought of a singlet in $C^2 \otimes C^2$  where the first (second) factor resides in $A$ ($A^c$). Raising an index, such a line can equally be read as the ``identity'' isomorphism between these two qubits.  For some (deep) reason the lines must keep of order a Planck distance apart. This picture is a nice starting point for entangling $A$ and $A^c$, but can readily be enhanced. Instead of mere disjoint arcs we could consider a quantum circuit---again with a Planck-scale density restriction---joining $A$ to $A^c$. In this circumstance there is again a maximal possible entanglement, $-\tr\rho\log\rho$, between input and output of the circuit. The advantage of replacing singlet strands with a general quantum circuit is that more general  entanglement structures can thus be realized and these may be required to model the physics in the holographic dual.  Recent work \cite{Cui:2015pla,Hastings} shows that entanglement depends on number-theoretic properties of the Hilbert space dimensions in such a network. This feature may be useful in model building.

Alternatively, let us look at our bit threads with the eyes of a harmonic analyst.  What could they  mean? Space is not really packed full of Planck-spaced strings, they must be surrogates for something more fundamental: vibrational modes. How can we set this up?  Let $H$ be the Hilbert space of divergenceless flows through the bulk from $A^c$ to $A$ normalized at infinity (or equivalently closed $d-1$-forms evaluating $\pm1$ on tangent $d-1$ planes to $A^c$ ($A$) and zero on all perpendicular planes.) The Riemann-Laplace operator is elliptic with these (Dirichlet) boundary conditions and we may use its eigen-fields as a basis for $H$. It seems a reasonable ansatz to imagine separating variables (as in a Sturm-Liouville problem) finding these eigen-fields counted  by the corresponding eigenfunctions of the Riemann-Laplace operator on functions  defined on a least-area hypersurface $Y= m(A)$. We will assume this. There is a beautiful discovery of Herman Weyl's \cite{Weyl} that the number of harmonic eigenfunctions up to a certain wavelength  is very nearly the $d-1$ volume of $Y$ measured in units of that wavelength, and further the next order (monotone) correction is given by the average scalar curvature $\kappa$ of $Y$. From \cite{MR2346276} we extract the following asymptotic information by integrating their pointwise estimates over $Y$. To match their notation set $d-1=n$.

Let $s_n$ be the volume of the unit ball in $R^n$, and define $N(\lambda)$ as the number of eigenvalues of the Laplacian on $Y$ whose square root is less than $\lambda$. Then
\begin{equation}
N(\lambda) = \frac{Vol(Y)s_n}{(2\pi)^n} \lambda^n + R(\lambda)\,, 
\end{equation}
where the error term $R(\lambda)$ has the form:
\begin{equation}
R(\lambda) = \frac{\kappa}{6\Gamma(\frac n2)}\lambda^{n-2} + c_{n-4}\lambda^{n-4} + \cdots+c_{(1,2)}\lambda^{(1,2)}+R^{\rm osc}(\lambda)\,,
\end{equation}
where $(1,2)=1$ if $n$ is odd and $2$ if $n$ is even, and $R^{\rm osc}(\lambda)$ is the oscillatory piece of the expansion. For generic metrics there is a lower bound, $R^{\rm osc}(\lambda)=\Omega(\lambda^{(n-1)/2})$, but for certain metrics, e.g.\ the round sphere, $R^{\rm osc}(\lambda)$ is much larger; there, $\limsup_{\lambda\to\infty}R^{\rm osc}(\lambda)/\lambda^{n-1}>0$ (but finite).

Weyl's Law, as it is called, allows us to regularize the Hilbert space $H$ to finite dimensions by applying  a Planck-scale cutoff  $\lambda$ to the basis. Thus $H$ becomes the Hilbert space of vibrational modes of (divergenceless vector) fields running from $A^c$ to $A$.  The dimension of $H$ is essentially the area of $m(A)$ in Planck units. To interpret this dimension via the RT formula it is natural to pass to the fermionic Fock space $F$ of $H$. $F$ has dimension $2^{\dim H}$, which is convenient because then entropy of a random vector in a Hilbert space may be normalized to  $\log_2(\text{dim of Hilbert space})$. (Technically it requires infinite information to specify a vector exactly but the above formula works perfectly to distinguish basis vector within a fixed basis and has the correct formal properties with respect to disjoint union of physical systems/tensor product of Hilbert spaces.) Now $\log_2(\dim F) = \dim H  \approx  \area(m(A))$ so RT is telling us that the EE between $A$ and $A^c$ is approximately the entropy of a fermionic system generated by the vibrational modes (up to Planck cutoff) of the divergenceless vector fields  propagating from $A^c$ to $A$. A many-body state here expresses the entanglement in the holographic dual. This, finally, is an apples-to-apples, entropy-to-entropy, comparison. 

There is a final hint to follow up on. In the simplest gravity theories where the Einstein-Hilbert action is extended to higher-order terms in the curvature tensor, namely Gauss-Bonnet gravity, it is believed \cite{Hung:2011xb,deBoer:2011wk} that the area on the gravity side of the RT formula area should be replaced with $\int dA( 1+\alpha R)$ (and minimization should be done with respect to this functional).  Weyl's law, available for use in this harmonic viewpoint,  suggests an explanation.  If what must be counted are not units of area but harmonic functions, then the appearance of this new functional looks natural. In the functional above, the scalar curvature $R$ is measured not at the Planck scale but at an intermediate string scale a dozen orders of magnitude larger, so to really claim that Weyl's law can explain RT for Gauss-Bonnet gravity, the appearance of the string scale needs to make sense. The bulk geometry, even if defined up to the Planck scale, is the base manifold of a string field theory, so it is reflected in that theory only insofar as  it can be probed by strings.  Curvatures between the Planck scale and the string scale  may  largely decouple from the bulk theory. Perhaps this is the start of an explanation.

\section{Open questions}\label{sec:discussion}

We close with a series of open questions concerning the bit-thread picture of holographic entanglement. We've already touched on a few of these in the previous sections.

\subsection{Constraints on holographic states}

It remains to find flow-based proofs of the monogamy of mutual information (MMI) inequality \eqref{MMI} \cite{Hayden:2011ag,Headrick:2013zda} and its generalizations to more than three regions \cite{Bao:2015bfa}. We showed in subsection \ref{sec:threeregions} that MMI \emph{cannot} be proved using just the nesting property and other basic properties of flows. Therefore, this inequality reflects some other property of flows, currently unknown to us. This is not just a technical problem. The meaning of MMI and its generalizations---what do they tell us about the special entanglement structure of holographic states?---remains obscure. Since the flows (or bit threads) provide a visual representation of this entanglement structure, it seems likely that a flow-based proof of the inequalities would help us to understand the meaning of MMI.

This is closely related to a second question, touched upon in subsection \ref{sec:bitthreads2}: What types of quantum states admit a bit-thread representation? In other words, how (if at all) are holographic states constrained by the fact that they admit such a representation?

\subsection{Generalizations of Ryu-Takayanagi}

The Ryu-Takayanagi formula, and therefore its max flow formulation and the accompanying bit-thread picture developed in this paper, apply within a certain regime: the bulk should be governed by classical Einstein gravity, it should be in a static state, and the region $A$ should lie on a constant-time slice of the boundary. Thus the RT formula can be generalized in at least three directions, by relaxing the static, Einstein, and classical conditions respectively.

The covariant generalization of RT, the Hubeny-Rangamani-Takayanagi (HRT) formula \cite{Hubeny:2007xt}, replaces the minimal surface with an extremal (spacelike codimension-2) surface. The flow version of HRT will be elucidated in forthcoming work \cite{covariantflows}.

Higher-derivative corrections to Einstein gravity include, for example, $\alpha'$ corrections in string-theory realizations of holography. In the special case of Lovelock gravity, it is believed that the RT formula is corrected by replacing the area functional by a functional that is essentially the lower-order Lovelock functional on the surface \cite{Hung:2011xb,deBoer:2011wk}. For example, a Gauss-Bonnet correction to the gravitational action adds an Einstein-Hilbert term to the area functional; the total is then minimized to give the entanglement entropy. In subsection \ref{sec:Weyl}, using the Weyl law for the distribution of Laplacian eigenvalues on a manifold, we gave a possible generalization of max flow-min cut (MFMC) that would naturally incorporate such corrections. For more general higher-derivative corrections, the appropriate generalization of the RT formula is not known (see however the attempt \cite{Dong:2013qoa}). Going even farther afield, one can consider EEs in duals of higher-spins fields, where the RT formula appears to be replaced by some sort of bulk Wilson line \cite{deBoer:2013vca,Ammon:2013hba,Castro:2014mza}, and ask what the analogue of the max flow would be in that case.

\begin{figure}[tbp]
\centering
\includegraphics[width=0.6\textwidth]{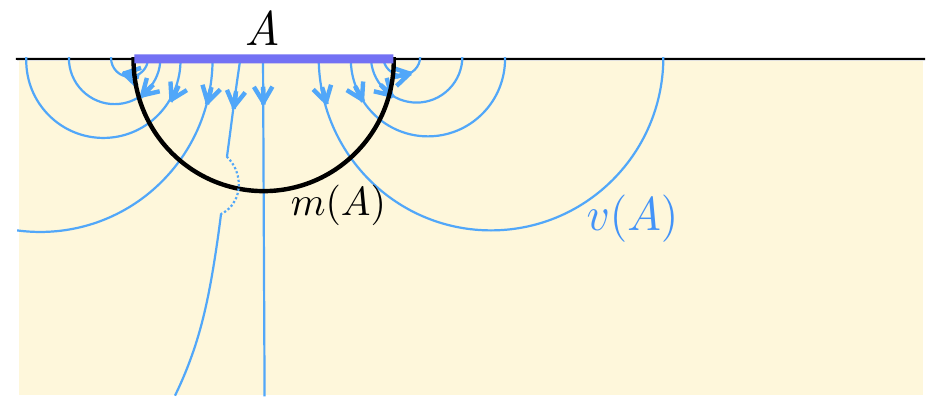}
\caption{\label{fig:quantum}
Conjectured bit-thread representation of quantum corrections to entanglement entropy. In addition to the order-$N^2$ threads that are required to be continuous, there are of order $N^0$ threads that can jump from one part of the bulk space to another, representing entanglement of bulk quantum fields. The ones that jump across the minimal surface $m(A)$ contribute to the max flow on $A$ and therefore to $S(A)$. Such a jump is shown as the dotted part of one of the threads in the figure.
}
\end{figure}

Quantum corrections are controlled by $G_{\rm N}\hbar\sim 1/N^2$ (in the usual large-$N$ parlance). Thus the leading perturbative correction is of order $N^0$ (versus $N^2$ for the leading term). At this level, one has to be more careful than we have been in distinguishing between the continuous flows and the discrete threads. Presumably one also has to take into account the fact that the bulk metric is undergoing quantum fluctuations. Perhaps more interesting, however, is the fact that entanglement of the bulk quantum fields contributes to the boundary EE \cite{Faulkner:2013ana}. Specifically, at order $N^0$ the contribution is given by the EE of the fields in the bulk spatial region $r(A)$ defined by the homology constraint, $\partial r(A) = A\cup m(A)$ (the so-called ``entanglement wedge'' \cite{Headrick:2014cta}). (At this order, the bulk fields should be treated as free.) How can we reproduce this contribution using bit threads? A natural guess is that the effect of bulk entanglement is to allow the threads to jump from one part of the bulk to another. (These threads can be thought of as traversing Planck-sized wormholes connecting distant parts of the bulk, perhaps in the spirit of Maldacena and Susskind's ``ER = EPR'' proposal \cite{Maldacena:2013xja}.) Some threads would thus be able to ``tunnel'' across the bottleneck $m(A)$, increasing the maximum number that can leave $A$ and thereby increasing $S(A)$ (see fig.\ \ref{fig:quantum}). As usual, there would not be a single thread configuration, but rather a collection of them, accounting for the fact that the entanglement structure of the bulk quantum fields does not just consist of a set of pairwise entangled localized bits. It would be interesting to make this rather speculative idea more concrete and quantitative.

\subsection{Emergent geometry}

Finally, let us return to the question posed in the first paragraph of this paper: What role does holographic entanglement play in the emergence of spacetime? The reconstruction of the spacetime metric from the EEs of boundary regions (and a related quantity called ``entwinement'') has been investigated recently by Czech and collaborators \cite{Balasubramanian:2014sra,Czech:2014ppa,Czech:2015qta}. The bit threads suggest a novel way to approach this question, which goes beyond the data contained in the boundary EEs. Suppose that the bulk space is initially given as a topological manifold, without a metric. The entanglement structure of the field theory would be expressed as a collection of thread configurations on this manifold. Saying that the threads have a cross-sectional area of 4 Planck areas then endows the manifold with a geometry. Specifically, the metric can be defined as the smallest that permits all the given thread configurations. In other words, space is propped open by the threads.

We can make the above picture mathematically precise. First, in the absence of a metric, a thread configuration is represented not by a divergenceless vector field $v$ but rather by a closed $(d-1)$-form $w$. Given a metric, the form can be converted into a vector field using the Hodge star: $v^\mu = g^{\mu\nu}(\star w)_\nu$. Therefore, in terms of $w$, the constraint $|v|\le1/(4G_{\rm N})$ on the density of threads translates to
\begin{equation}\label{constraint}
\frac1{(4G_{\rm N})^2} \ge v^2=w^2 = \frac1{(d-1)!}g^{\mu_1\nu_1}\cdots g^{\mu_{d-1}\nu_{d-1}}w_{\mu_1\cdots\mu_{d-1}}w_{\nu_1\cdots\nu_{d-1}}\,.
\end{equation}
Whereas before the metric was taken as fixed and this inequality was viewed as a constraint on $v$, we now take the components of $w$ as fixed and view it as a constraint on the metric. Given a collection of forms $W=\{w\}$, we let $g_{\mu\nu}$ at each point be the smallest-determinant positive symmetric matrix satisfying \eqref{constraint} for all $w\in W$. (This is itself a convex optimization problem.) This is what we meant above by the phrase, ``the smallest metric that permits all the given thread configurations''.

Of course, this construction leaves many questions unanswered. First and foremost is the question of where this primordial collection $W$ of thread configurations comes from, and how it is related to the entanglement structure of the dual field theory. At a minimum it should reproduce the EE of all boundary spatial regions. This requires $W$ to include at least one max flow for each region. However, in principle $W$ could include many more configurations than that. Including enough thread configurations would allow the entire bulk metric to be determined via the above construction, including inside so-called ``entanglement shadow'' regions, where RT minimal surfaces do not pass \cite{Engelhardt:2013tra,Balasubramanian:2014sra,Freivogel:2014lja}. A fundamental technical question is what property must $W$ satisfy for the resulting metric $g_{\mu\nu}$ to be smooth (or even continuous). Knowing this would help us decide which forms $w$ to place in $W$. 

More ambitiously, one would want not just the geometry but also the topology of the bulk to be emergent, i.e.\ for the manifold itself to be determined by a set of threads living in some more abstract space. We leave the exploration of these questions to future work.

\acknowledgments

We would like to thank N. Bao, R. Bryant, B. Czech, S. Hartnoll, A. Lawrence, H. Liu, R. Myers, J. Preskill, X. Qi, M. Rangamani, E. Tonni, and B. Yoshida, and especially P. Doyle, P. Hayden, V. Hubeny, F. Morgan, and J. Sullivan, for very useful discussions.

This work was initiated at the Aspen Center for Physics, which is supported by National Science Foundation grant PHY-1066293, and developed at the Simons Symposium on Quantum Entanglement and at the Kavli Institute for Theoretical Physics, which is supported in part by the National Science Foundation under Grant No.\ NSF PHY11-25915. The work of M.H. was supported in part by the National Science Foundation under Career Award No.\ PHY-1053842.

\appendix

\section{Max-flow min-cut on manifolds}\label{sec:mfmc}

We put the max flow/min cut principle into a larger context.  For us the primary context, a divergenceless flow and least area hypersurfaces dual to it, can be extended in two different directions (and historically each direction had separate origins \cite{FF, EFS, HL, Federer69}).  One direction, ``calibrations,'' considers pairs (closed $p$-form, $p$-submanifold) which are, in a sense, minimal with respect to each other \cite{WikiCG}.  The second direction is to pass from smooth flows to discrete networks and treats discrete analogs of flux and cross-sectional area.  Both extensions are of interest in gravity.  For example, $p = d-2$ is a very natural case when considering surfaces in a $d$-dimensional space-time bounding a cut dividing the constant time holographic dual into two pieces $A$ and $B$.  (The math literature does not seem to have considered calibrated geometry in mixed signature, so the standard results, reviewed here, are for Euclidean signature and will require some re-consideration in the Lorentzian case.)  Discretization is also a natural direction if one wishes to interpret the flow as quantum mechanical entanglement between $A$ and $B$, in which case the flow might be replaced with a quantum circuit \cite{Cui:2015pla}.

\subsection{Relation to calibrations}

On a Riemannian $d$-manifold we may use the metric to convert a $p$-form $w$ to a ($d-p$)-vector (field) $v$.  In the special case that $p=d-1$, the condition that $w$ is closed, $d w=0$, is equivalent to $\operatorname{div}(v) = \nabla\cdot(v) = 0$.  While our primary interest has been a divergenceless vector field $v$ and its flux through hypersurfaces, i.e., the case $p=d-1$, we will make some general statements now for $p$-forms $p<d$ and their integrations over $p$-submanifolds, or integral currents, and return later to the special case $p=d-1$. An integral $p$-current is essentially an oriented rectifiable set of dimension $p$ with integral weights, thought of as a functional on $p$-forms via integration. If this functional annihilates exact forms the underlying rectifiable set is thought of as a closed ``singular $p$-submanifold''. Regularity theory discusses how bad the singularities must be to realize the infimum of $p$-area within the integral homology class $\alpha\in H_p(M;\Z)/$torsion, where $M$ is the ambient manifold. Integral currents have the merit that their underlying rectifiable sets have good compactness properties, guaranteeing the existence of minimizers.

A $p$-calibration is a closed real-valued $p$-form $w$ so that at every point $\left|w(v)\right|\leq 1$, where $v$ is a unit norm \emph{geometric} $p$-vector.  The word \emph{geometric} means an element of the Grassmannian $G_{p,d}$ rather than an element of formal $p$-planes, $\wedge^p\left(T_M\right)$, the $p^\text{th}$ exterior power of the tangent bundle.  The norm on $p$ planes, geometric or otherwise, is induced by taking $\{e_{i_1}\wedge\cdots\wedge e_{i_p}\}$ as an orthonormal basis, $i_1<i_2<\cdots<i_p$, and $\{e_i\}$ an orthonormal basis of $T_M$.  Thus norm $||w||$ is an $L^\infty$-norm.  Since it is pointwise and is defined only through its evaluation on geometric $p$-vectors, it could be denoted $||\;||^\infty_g$ but we sometimes use $||\;||$.  As an example in Euclidean $4$-space $E^4$, $w = d_{x_1}\wedge d_{x_2} + d_{x_3}\wedge d_{x_4}$ is a calibration even though $w\left(\frac{1}{\sqrt{2}} e_1\wedge e_2 + \frac{1}{\sqrt{2}} e_3\wedge e_4\right) = \sqrt{2}$.  The point is that $\left(\frac{1}{\sqrt{2}} e_1\wedge e_2 + \frac{1}{\sqrt{2}} e_3\wedge e_4\right)$, although of unit norm, is \emph{not} a \emph{geometric} $2$-vector (i.e., plane).

One says an oriented $p$-submanifold $P$ is \emph{calibrated} by $w$ if for all $x\in P$, $w(v_x) = 1$ for $v_x$ the oriented unit $p$-area $p$-vector tangent to $P$ at $x$.

If $w$ calibrates $P$, and $P$ has finite $p$-area, then $P$ absolutely minimizes area (let us drop the $p$- in $p$-area) for any (weighted) submanifold $Q$ (real-)homologous to $P$.  (Proof: $\operatorname{area}(P) = \int_P 1\;da = \int_P w = \int_Q w \leq \int_Q 1\; da = \operatorname{area}(Q)$).  There is a converse to this but it is surprisingly weak.  One must absorb some interesting counterexamples to appreciate why.

First, to understand the role of real coefficients, consider an example in the product of circle and a $2$-sphere $S^1\times S^2$.  Embedded a loop $\gamma$ which represents twice the generator of $H_1\left(S^1\times S^2;Z\right)$.  Now multiply the product metric on $S^1\times S^2$ by a conformal factor $e^{-\lambda(x)}$ where the function $\lambda(x)$ becomes extremely large near $\gamma$.  For such a metric,
\begin{equation}\label{eq:ast}
\operatorname{length}(\gamma) \ll \operatorname{inf}\left(\operatorname{length}\left(\alpha^\prime\right)\right)\tag{$\ast$}
\end{equation}
where $\alpha^\prime$ is any loop representing the generator of $H_1\left(S^1\times S^2;\Z\right)$.  Now let $\alpha$ be a shortest closed geodesic realizing the generator; $\operatorname{length}(\alpha)=\operatorname{inf}\left(\operatorname{length}\left(\alpha^\prime\right)\right)$. For generic metrics, $\alpha$ will be simple (non self-intersecting) so we assume this.  $\alpha$ cannot be calibrated by any $1$-form $w$, otherwise
\[\operatorname{length}(\gamma) \geq \int_\gamma w = 2\int_\alpha w = 2\int_\alpha 1 = 2\operatorname{length}(\alpha),\]
contradicting (\ref{eq:ast}).

There is a more interesting expression of this phenomenon (due to Young, White, and Morgan \cite{Young, White, Morgan}).  For any integer $n\geq 2$, there are smooth simple closed curves $\beta$ in $E^4$ so that the minimal area of a bounding surface for the multiply weighted $n\beta$ is less than $n$ times the minimal area of a bounding surface for $\beta$.

These examples show that one should only attempt to calibrate submanifolds which minimize area within their \emph{real} homology class.  Fortunately when $p=d-1$, regularity theorems \cite{Federer74, MR809969} say that real coefficient $p$-area minimizers (which exist for any $p$ as integral currents) are in fact linear combinations of smooth submanifolds for $d\leq 7$, $p=d-1$, and with integral weights become minimizers for integral classes.  For $d\geq 8$, the Simons phenomenon appears \cite{BDG,MR0216387}, and for certain metrics codimension one minimizers will be singular. Furthermore\footnote{This was pointed out to us in \cite{Chodosh}.} when $d=8$ it is proved in \cite{MR1243523}, using the local analysis of \cite{BDG} and \cite{MR809969}, that for generic metrics codimension one minimizers will be smooth manifolds. For $d>8$, it is conjectured in \cite{MR1243523}, but remains open, that for generic metrics, codimension one minimizers are smooth submanifolds. It is known \cite{MR809969} that the singularities of such minimizers may be assumed to have Hausdorff dimension $\le p-7$.

We now address the question of whether a smooth, real, area minimizer $P$ can be calibrated.  The most basic context is to work within a closed manifold $M^d$, but in many applications (e.g., gravity), one is interested in noncompact $M$ and $P$.  In this case, ``least area'' means that $P$ cannot be modified in a compact region to a homologous $P^\prime$ with lower area. (Note: In discussing infinite area surfaces this seems to be the best we can do, whereas in the dual context of flows, as we saw in section \ref{sec:reformulation2}, absence of an augmenting path provides a sharper notion.) Our discussion applies to both cases although our notation generally presumes the former.

\subsection{Regularity}

Regularity questions are central, and a good general introduction is \cite{MR2455580}.  We make no attempt here at rigor but merely outline some important ideas.  The broadest generalization of submanifold is \emph{current}.  We need the following notations:
\begin{itemize}
    \item $M$ is a complete Riemannian manifold.
    \item $\wedge^p(M)$ are $p$-forms on $M$, initially smooth but this property may be lost as limits are taken.
    \item $\mathcal{C} = \left(\wedge^p\right)^\ast$ is the dual space called \emph{currents}.
    \item $I_p$ denotes integral $p$-currents. These are the functionals obtained by integrating $\wedge^p$ over an oriented rectifiable $R$ with integral weights. A rectifiable set $R$ by definition is the disjoint union of a countable number of oriented Lipshitz charts from $\R^p$ together with a ``negligible'' piece of $p$-Hausdorff measure 0. This is a nice type of current on which one can integrate forms, but generally less regular than a submanifold.  (For example, general currents also include analogs of $\delta^\prime$, derivatives of delta functions, along a submanifold; these are not integral currents. Oriented submanifolds with low dimensional singularities can still be integral currents.)
\end{itemize}

Motivated by Stokes' theorem, an integral current $c$ is called \emph{closed} if it annihilates exact forms $c\left(d\alpha^{p-1}\right) = 0$. In this case $c$ defines a homology class $[c]\in H_p(M;\Z)/$torsion. Closed integral currents (actually their underlying rectifiable set $R$) have a well-defined $p$-area and a minimizer $c'$ is known to exist in the class $[c]$.

For $p\leq d-2$ there are smooth submanifolds $P$ absolutely minimizing area (even among closed integral currents) in their real homology class, for which no calibrating form can be a smooth or even continuous section of the $p$-form bundle
\[\xymatrix{\wedge^pT^* M \ar@{-}[r] & E \ar[d] \\ & M}.\]
This problem exists even locally in a neighborhood of $P$ \cite{Zhang}.  (While smooth calibrations do not always exist, it is possible that they may exist for generic metrics for all dimensions and codimensions.)

The regularity problem is not too severe for $p=d-1$.  For $d\leq 7$, minimizers for classes in $H_{d-1}(M;R)$ exist and are supported on a smooth submanifold $P$ which (by reweighting) are minimizers for the corresponding class in $H_{d-1}(M;Z)/$torsion.  Furthermore $P$ is calibrated by a smooth $1$-form $w$.  In the dual language of vector fields, if $P$ is least area in its class in $H_{d-1}(M;Z)/$torsion, there exists a vector field $v$, with $||v||^\infty=1$, $\nabla\cdot v = 0$, and $\operatorname{flux}_P(v) = \int_P da\; v_\text{normal} = \operatorname{area} P$, i.e., $v$ is unit normal to $P$.  
As remarked earlier, these statements hold for $d=8$, for generic metrics, and have also been conjectured for $d>8$ again for generic metrics \cite{MR1243523}. This is the smooth max flow/min cut (MF/MC) theorem:

\begin{theorem}\label{theoremA1}
$P^p\subset M^d$, $p=d-1$, is a submanifold in an oriented ambient manifold minimizing area in its integral homology class.  Then there is a smooth divergenceless vector field $v$ with the pointwise norm $||v(x)||\leq 1$, $x\in M$, $\operatorname{flux}_P(v) = \operatorname{area}(P)$.  (Clearly $p$-$\operatorname{area}(P)$ serves as an upper bound to flux.)  Furthermore the set of $v$ satisfying the theorem is clearly convex: $tv + (1-t)v^\prime$, $0\le t\le1$, will be a max flow whenever $v$ and $v^\prime$ are.
\end{theorem}

\paragraph{\it{Non-closed manifolds.}}
If $M$ has a boundary $\partial M$ with the weak convexity property that its mean curvature vector is never pointing outward, one may consider $\left(P^p, \partial P^p\right)\subset\left(M^d,\partial M^d\right)$, $p=d-1$, minimizing area with respect to free boundary conditions, for a class in $H_p(M,\partial M; Z)$.  Such a hypersurface is again associated to a vector field $v$ with $\nabla\cdot v=0$, $||v(x)||\leq 1$, $x\in M$, and $\operatorname{flux}_P(v)=\operatorname{area} P$. Similarly if $M$ is a complete Riemannian manifold and $\rho\in H_p(M,\infty;Z)$ then if $P$ is properly embedded in $M$, $[P] = \rho$, then there exists a vector field $v$ on $M$ with $\nabla\cdot v=0$, $||v(x)||\leq 1$, $x\in M$, $v(x)$ normal to $P$ and $||v(x)||=1$, $x\in P$.

Again, as in the closed case, all minimizers will be manifolds for $d\leq 7$.

\subsection{Proofs}

Given the existence of a manifold minimizer $P$, theorem \ref{theoremA1} can be extracted from \cite{Federer74}.  The regularity of the minimizer has a long history culminating in \cite{BDG, MR809969}.  We did not find fully generic statements for the bounded case in the math literature, but the arguments should be similar to the closed case.  Similarly the complete case might not be stated in the literature but should follow from the bounded case by exhausting $M$ by compact submanifolds $\{M_i,i=1,2,3,\ldots\}$.  $M_i\subset\operatorname{int} M_{i+1}$ and $\cup_{i=1}^\infty M_i = M$, with $\partial M_i$ of non-negative mean curvature. (One approach to constructing such an exhibition is to let $M_i$ solve the isoparametric problem: $M_i\subset M$ with $\operatorname{volume}\left(M_i\right) = v_i$, $v_i<v_{i+1}$, generic values approaching $\infty$.)  Use uniform curvature bounds to take a smooth Gromov limit of minimizers $P_i\subset M_i$.

There is a second, more computationally oriented approach to these results \cite{Sullivan}, which is to obtain then a limit from discrete graph-theoretic versions of MF/MC which begin with Menger's theorem \cite{Menger}.  (In an unoriented graph the number of edge disjoint paths between any vertices $s$ and $t$ is equal to the number $\left|E^\prime\right|$ of edges $E^\prime\subset E$ which must be deleted to separate $s$ from $t$.)  Discrete statements of MF/MC for oriented and weighted graphs are given on Wikipedia; we establish some terminology.  $G(V,E)$ is a graph with oriented edges $E = \{e\}$.  Let $w_i(e_i)\in\mathbb{R}^+\cup 0$ be the weights or \emph{capacity}, $f(e)\in\mathbb{R}^+\cup 0$ be the flow.  The goal is to maximize the flow from a source vertex $s$ to a sink vertex $t$.  The unoriented case may be mapped to the oriented case by replacing each edge $e$ with two $e_+$ and $e_-$ oppositely oriented edges of equal capacity.

\begin{enumerate}
    \item\label{itm:cc} \emph{capacity constraint}: $f(e)\leq w(e)$.
    \item\label{itm:div} \emph{$\operatorname{div} = 0$}: For $v\in V$, $v\neq s$, $t$, $\sum_{e\text{ into }v} f(e) = \sum_{e\text{ out }v} f(e)$.
    \item \emph{flux}: $\operatorname{flux}(f) = \sum_{e\text{ out }s} f(e) = \sum_{e\text{ into }t} f(e)$.
\end{enumerate}

Given a flow $f$, the \emph{residual network} $G_f(V,E)$ has the same vertices $V$ but with new weights $w_f(e) = w(e) - f(e)$; by convention, edges $e$ with $w_f(e) = 0$ are deleted.  The Ford-Fulkerson algorithm (FF) begins with any initial \emph{feasible} flow $f$, i.e., one obeying (\ref{itm:cc}) and (\ref{itm:div}), perhaps the identically zero flow.  As long as there is \emph{some} path $\gamma$ in $G_f$ from $s$ to $t$, the algorithm sends the maximum permissible additional flow from $s$ to $t$ along $\gamma$.  Such a $\gamma$ is called an \emph{augmenting path}.  If all weights are rational numbers, FF terminates with a MF in finitely many steps.  With general real weights, FF still converges to a MF if the time steps are indexed over the ordinals.  (Of course indexing over the ordinals violates the usual rules for an \emph{algorithm}.  With general real weights, other algorithms based on linear programming converge to a MF in polynomial time.)  Two important references for MF/MC are the 1956 papers \cite{EFS} and \cite{FF}.

For us, the essential consequence of FF is that if $f$ is a feasible flow on a graph $(G;s,t)$, setting $\Delta = (\operatorname{max flow}(G;s,t) - \operatorname{flux}(f))$, then there is a max flow $f^{\prime\prime} = f + f^\prime$ for $(G;s,t)$ where the pointwise norm of $f^\prime$ satisfies $\left|f^\prime\right|^\infty\leq\Delta$.  We call this the \emph{extension principle}.  Roughly this says that FF is a greedy algorithm.

\subsection{Continuity}\label{continuity}

\paragraph{\it{Graph case.}}
We previously observed that max flows for any fixed problem (in both the graph and Riemannian context) form a convex set $X$.  But now suppose there is a continuous variation in the problems considered. First the graph case: Let $(G;s,t)(\tau)$ have weights which are a function of a parameter $\tau$ from any reasonable space; for notational convenience, take $\tau\in\mathbb{R}$, the reals.  Our notation for a parametric MFMC problem will simply be $G_\tau$.  We now show that it is possible to select max flows $v_\tau$ continuously as $\tau$ varies.  This is in contrast to the physical position of a minimal cut edges $E_\text{cut}$ which may jump as one constriction may evolve to be suddenly more acute than another, as in an energy crossing in a first order phase transition.

Let $K_\tau$ be the convex set of MF solutions at ``time'' $\tau$.  To establish the existence of a continuous selection $v_\tau\subset K_\tau$, we show something stronger.  We thank Peter Doyle for a discussion on this lemma.

\begin{lemma}\label{lem:continuity}
Given a continuous family of weights, $w_{i,\tau}$, on the edges of an (oriented or unoriented) graph with source and sink $(G;s,t)$ the convex set of max flow solutions $K_\tau$ evolves (Lipschitz) continuously in the Hausdorff topology on closed subsets of Euclidean space.
\end{lemma}

\begin{proof}
We must show that for every $\epsilon > 0$ there is a $\delta > 0$ so that $K_{\tau\pm\delta}\subset \mathfrak{N}_\epsilon\left(K_\tau\right)$, where $\mathfrak{N}_\epsilon(\;)$ denotes the $\epsilon$-neighborhood in the sup norm Hausdorff topology on MF solutions.  In figures below, \ref{fig:a} is a discontinuous change (does not occur) whereas \ref{fig:b} and \ref{fig:c} are continuous and may occur.

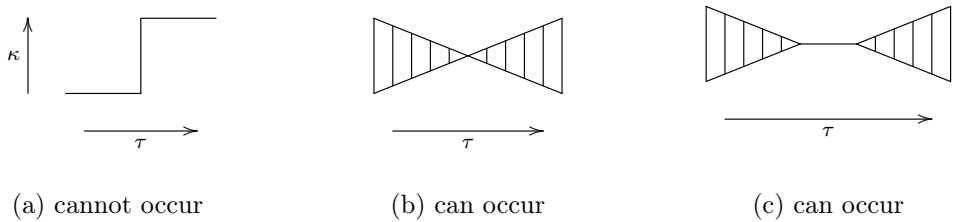
\begin{figure}[hbpt]
    \centering
    \begin{subfigure}[b]{0.3\textwidth}
        \[\begin{xy}<5mm,0mm>:
        \ar^\kappa (0,0);(0,2)
        \ar@{-} (1,0);(3,0)
        \ar@{-} (3,0);(3,2)
        \ar@{-} (3,2);(5,2)
        \ar_\tau (1.5,-1);(4.5,-1)
        \end{xy}\]
        \caption{cannot occur}
        \label{fig:a}
    \end{subfigure}
    \begin{subfigure}[b]{0.3\textwidth}
        \[\begin{xy}<5mm,0mm>:
        \ar@{-} (0,0);(0,2)
        \ar@{-} (0,2);(5,0)
        \ar@{-} (5,0);(5,2)
        \ar@{-} (5,2);(0,0)
        \ar_\tau (0.5,-1);(4.5,-1)
        \ar@{-} (0.5,1.8);(0.5,0.2)
        \ar@{-} (1,1.6);(1,0.4)
        \ar@{-} (1.5,1.4);(1.5,0.6)
        \ar@{-} (2,1.2);(2,0.8)
        \ar@{-} (3,1.2);(3,0.8)
        \ar@{-} (3.5,1.4);(3.5,0.6)
        \ar@{-} (4,1.6);(4,0.4)
        \ar@{-} (4.5,1.8);(4.5,0.2)
        \end{xy}\]
        \caption{can occur}
        \label{fig:b}
    \end{subfigure}
    \begin{subfigure}[b]{0.3\textwidth}
        \[\begin{xy}<5mm,0mm>:
        \ar@{-} (0,0);(0,2)
        \ar@{-} (0,2);(2.5,1)
        \ar@{-} (2.5,1);(0,0)
        \ar@{-} (2.5,1);(4,1)
        \ar@{-} (4,1);(6.5,2)
        \ar@{-} (6.5,2);(6.5,0)
        \ar@{-} (6.5,0);(4,1)
        \ar_\tau (0.5,-1);(6,-1)
        \ar@{-} (0.5,1.8);(0.5,0.2)
        \ar@{-} (1,1.6);(1,0.4)
        \ar@{-} (1.5,1.4);(1.5,0.6)
        \ar@{-} (2,1.2);(2,0.8)
        \ar@{-} (4.5,1.2);(4.5,0.8)
        \ar@{-} (5,1.4);(5,0.6)
        \ar@{-} (5.5,1.6);(5.5,0.4)
        \ar@{-} (6,1.8);(6,0.2)
        \end{xy}\]
        \caption{can occur}
        \label{fig:c}
    \end{subfigure}
    \caption{Vertical fibers are $K_\tau$}
\end{figure}

As $\tau$ perturbs the weights, say from $\tau=0$, the original flow $f$ may violate the new constraints, but $f$ can be continuously scaled back by constant factor $\kappa = (1-\epsilon)$, so that $\kappa f$ obeys the constraints $\{w_{i,\delta^\prime}\}$, $\delta^\prime\in[-\delta,\delta]$.  The new flow $\kappa f$ is not, in general, a MF on any $G_{\delta^\prime}$, but for every $\epsilon > 0$, there exists a $\delta > 0$, so that $\Delta = \left(MF\left(G_{\delta^\prime}; s,t\right) - \operatorname{flux} \kappa f\right)$ satisfies $\Delta\leq (1.1)\epsilon F$, where $F = \operatorname{flux}\left(MF\left(G_0;s,t\right)\right)$ and $\left|\delta^\prime\right| < \delta$.  (The factor $1.1$ is included to allow for a possible increase of weights, $w_{i,\delta^\prime} > w_i$, which can be made arbitrarily small by keeping $\delta$ small.)  Now by the extension principle above, we may find a max flow
\[f_{\delta^\prime} = \kappa f + f^\prime\text{, }||f^\prime||^\infty\leq\Delta.\]

This locates the new MF, $f_{\delta^\prime}$, close to the original $f$, establishing continuity of the solution set.  The linear nature of the estimates implies that the continuity is Lipschitz.
\end{proof}

\paragraph{\it{Riemannian case.}}
Next we consider the Riemannian analog of Lemma \ref{lem:continuity}.  Now instead of perturbing the weights of a graph, one perturbs a Riemannian metic $g_{ij}(\tau)$ and the boundary conditions $(A(\tau), B(\tau))$, where $A$ is a region on the sphere at infinity of the underlying time-slice of the holographic dual, at $B$ its complement.  In this case we must state a conjecture rather than a result.

\begin{conjecture}\label{conj:riemann}
Let $g_{ij}(\tau)$ be a continuous family of metrics on a Riemannian manifold $\left(M^d,\partial\right)$ where the boundary has a non-outward pointing mean curvature.  Let $C_\tau\subset\partial M$ be a $d-2$ manifold dividing some number: $0,1,2,\ldots$ of boundary components of $M$ into pieces $A$ and $B$, $A\cap B = C$.  Fix a class $x\in H_{d-1}(M,C;Z)/$torsion (note $C$ could be empty) and consider the max flow problem for flows dual to $x$ and
let $K$ be the (infinite dimensional) convex set of max flows.  We conjecture $K_\tau$ varies in a Lipschitz continuous fashion (as a subset of the Banach space of vector fields with the pointwise ($L^\infty$) norm) as a function of $\tau$.
\end{conjecture}

We do understand a bit about this conjecture, but unresolved analytical questions prevent us from stating it as a theorem.  Let us start with some of what we do know:

Let $\{P_i\}$ be the minimizers (for fixed $\tau$); if indeed several exist then they are all pairwise disjoint.  This is an important point because if we had two minimizers $P_1$ and $P_2$ with $P_1\cap P_2\neq\emptyset$, we could easily manufacture a $\tau$-parameter family contradicting the conjecture. We will here assume $P_1$ and $P_2$ are smooth submanifolds. Indeed if $P_1\cap P_2\neq\emptyset$, they actually must cross transversely at some point $p$, and in fact such $p$ are open dense in $P_1\cap P_2$.  This is a standard folk theorem of minimal surface theory.  The idea is that locally near any intersection point $q$, one may write $P_2$ as the graph over $P_1$ of a function $h$ obeying an elliptic P.D.E.  The elliptic maximum principle implies that $P_2$ lies on both sides of $P_1$ near $q$ and, assuming the metric real analytic, $P_1\cap P_2$ will be a real analytic variety; any point $p$ in its top strata will be a point of transverse intersection.  A local estimate around such a point $p$ can predict in terms of the dihedral angles and higher derivative a lower bound for how much ($d-1$)-area an exchange/corner rounding argument would save if the exchanged regions met $p$.  Given this bound, $\delta$, one may perturb $P_1$ and $P_2$ into transverse position so all area changes are $\ll\delta$.  Thus, we may assume $P_1$ and $P_2$ are smooth ($d-1$)-submanifolds and intersect transversely.  Now consider $M\setminus \left(P_1\cup P_2\right)$.  We assume $M$ is connected but claim $M\setminus \left(P_1\cup P_2\right)$ must be disconnected.  (Proof: Consider a short normal arc $\alpha$ to $P_1$ and disjoint from $P_2$ with endpoints $a$ and $b$. If $M\setminus\left(P_1\cup P_2\right)$ were connected, we could join $a$ and $b$ by a second arc $\beta$ in $M\setminus\left(P_1\cup P_2\right)$.  The union $\gamma = \alpha\cup\beta$ is dual to $P_1$ but disjoint from $P_2$, contradicting $P_1$ homologous to $P_2$.)  Now let $X$ be some connected component of $M\setminus\left(P_1\cup P_2\right)$, with $Y_i:=X\cap P_i$, $i=1,2$.  Without loss of generality, assume ($d-1$)-$\operatorname{area}\left(P_2\right)\leq$ ($d-1$)-$\operatorname{area}\left(P_1\right)$, otherwise switch notation $1 \leftrightarrow 2$.  Now exchange $P_1\leadsto \left(P_1\setminus Y_1\right) \cup Y_2 := P_1^\prime$;
\[\operatorname{area}\left(P_1^\prime\right)\leq\operatorname{area}\left(P_1\right),\]
and after rounding corners $P_1^\prime\leadsto P_1^{\prime\prime}$ we have
\[\operatorname{area}\left(P_1^{\prime\prime}\right)<\operatorname{area}\left(P_1\right),\]
a contradiction.

Note that if $P_1\cap P_2=\emptyset$, the preceding exchange argument merely swaps $P_1$ and $P_2$, there is no contradiction.  It is the ``partial swapping'' that produces the contradiction.

We now give an alternative argument from the point of view of flows which illustrates the power and simplicity of this viewpoint for certain questions.  Any MF $f$ must be the unit normal to \emph{any} MC, $P_i$, clearly impossible if $P_1$ and $P_2$ intersect transversely at any $p$, which as we have seen is a consequence of $P_1\cap P_2\neq\emptyset$.  From the flow perspective, the more technical exchange argument is unnecessary.

Suppose, contrary to what we just demonstrated, that there was a metric $g_{i,j}$ for which $P_1\cap P_2\neq\emptyset$.  We could build a $\tau$-deformation $g_{i,j}(\tau)$ so that for $\tau\leq 0$, $P_1$ was the unique MC and for $\tau\geq 0$ $P_2$ was the unique MC.  Again considering unit normals, such a family would imply discontinuous $\kappa_\tau$.  So we treat the impossibility of $P_1$ and $P_2$ intersecting as evidence for Conjecture \ref{conj:riemann}.

The difficulty in extending the graph proof (Lemma \ref{lem:continuity}) to \ref{conj:riemann} is twofold.  There is not a completely satisfactory technology for passing from combinatorial to smooth limits (despite the very interesting analysis of P.L. minimal surfaces in the MF/MC context \cite{Sullivan}) and a difficulty in scaling our pointwise estimate on the norm of $f^\prime$ (the augmenting flow)
\begin{equation}\label{eq:norm}
||f^\prime||^\infty\leq(1.1)\epsilon F
\end{equation}

The issue with this estimate is that as one refines any combinatorial approximation of the continuum, the weights on the edges of the approximation should approach zero, but for $\epsilon > 0$ fixed, the right-hand side of (\ref{eq:norm}) has a definite value.  To overcome this issue, a more \emph{local} discussion of the addition of augmenting paths would be in order.  Our present analysis permits the possibility that a small change in the metric near one site will reroute large amounts of flux at a faraway site.  To make a proof in the continuum this must be avoided, perhaps by reassigning augmenting flow among nearly parallel copies of an edge.  Such parallels should become available with the correct combinatorial approximation to the continuum, but the details are not obvious.

In the discrete case, we noted the concept of \emph{augmentation paths} \cite{WikiFF}, a path utilizing only the residual weight on each edge after an initial feasible solution has been proposed.  This concept has a natural generalization to the Riemannian setting.  If $v$, $||v_x||^\infty\leq 1$, is a vector field on $M$, there is a \emph{bundle of off-centered unit balls} $\{B_x\}$ defined by the property that $||w_x + v_x||\leq 1$ for all $w_x\in B_x$.  This is the analog of the residual graph. The analog of an augmentation path is a curve $C$ from $A$ to its complement (or, more generally, in the dual cycle to the homology class in which we are finding a min cut) whose unit tangent vector $t_x$ can be scaled to fit in $B_x$ all along the path, i.e.\ such that there exists a constant $\alpha>0$ satisfying $\alpha t_x\in B_x$ for all $x\in C$. Given such a path, one can construct an augmenting flow $w$ by taking a pencil of curves that fill out a small tube around $C$ and setting $w_x=\alpha t_x$ within the tube and 0 outside. One expects that, as in the network case, a version of the FF algorithm using such augmentation paths will be greedy, i.e.\ will converge to a max flow starting from any allowed flow. Such paths should play a prominent role in the proof of conjecture \ref{conj:riemann}.

\subsection{Nesting}\label{nesting}

The nesting property, defined in subsection \ref{MFMC}, states that, given two regions $A$, $B$ of the boundary of $M$, which without loss of generality we take to be disjoint, it is possible to find a flow $v(A,B)$ that is simultaneously a max flow with respect to $A$ and to their union $AB$.

In the network setting, this can easily be proved using the FF algorithm: Apply FF to find a max flow on $A$ starting from a max flow on $AB$; the greediness of FF says that this is possible. The only question is whether at the end we still have a max flow on $AB$. However, by definition the augmentations paths leave $A$, so they cannot reduce the flux on $AB$ (nor can they increase it, so in fact they must end on $B$).

As discussed in the previous subsection, a similar notion of augmentation path can be defined in the Riemannian setting, and one again expects to be able to define a greedy FF algorithm using this notion, although we have not worked out the crucial details.

\subsection{Approximating the continuum}

We close with a proposal for combinatorial approximation, as it may be useful for passing MF/MC arguments back and forth between discrete and continuous settings, and also for its possible applicability to the more general problem of regularizing quantum gravity.

The problem with a P.L. approximation to a Riemannian manifold is that, no matter how fine, the simplices have large mutual dihedral angles.  In the case of minimal surfaces, the best that one might hope for is a $C^0$-approximation, certainly not $C^1$.

Our proposal is to fill a Riemannian manifold up with needles of roughly constant length, to fill $\tau_U(M)$, the unit tangent bundle of $M$, as uniformly as possible.  One way to do this is to choose $0\ll \epsilon^{\prime\prime}\ll\epsilon^\prime\ll\epsilon\ll k$, $k$ the curvature scale.  Then take a maximal \emph{packing} of disjoint $\epsilon^{\prime\prime}$-balls in $M$.  This will be somewhat uniform in the sense that the radius $2\epsilon^{\prime\prime}$ balls with the same centers \emph{cover} $M$.  Now for the ``needles'' $N$: build an edge set $N$ of segments between all ball centers whose distance $\delta$ satisfies $\epsilon^\prime < \delta < (1.1)\epsilon^\prime$.  The resulting graph: vertices $=\{\epsilon^{\prime\prime}\text{-ball centers}\}$, edges$=N$, should be considered as $\epsilon^{\prime\prime}$, $\frac{\epsilon^{\prime\prime}}{\epsilon^\prime}$, $\epsilon^\prime$, $\frac{\epsilon^\prime}{\epsilon}$, and $\epsilon$ all approach zero.  This sequence of graphs appears capable of giving a reasonable representation of differential information. We illustrate this with the case at hand: flows.

Let us define maps $j$ and $k$ below:

\[\begin{xy}<5mm,0mm>:
(0,0)*=<130pt,50pt>[F]{\begin{array}{l}
\text{Smooth vector fields $v_x$} \\
\text{on $M$ with maximum} \\
\text{velocity $\leq 1$}
\end{array}}
,(15,0)*=<150pt,35pt>[F]{\begin{array}{l}
\text{Weighted arrows $\vec{n}$ on each} \\
\text{$n\in N$ with weight $\leq 1$}
\end{array}}
\ar^j (5,0.625);(9.25,0.625)
\ar_k (9.25,-0.625);(5,-0.625)
\end{xy}\]

To define $k$, fix a smooth ``averaging'' function $g: B_\epsilon\rightarrow R^+$ with $g\vert_{\partial B_\epsilon}\equiv 0$ and $\int_{B_\epsilon} d\text{vol}\;g = c_d$. The dimension-dependent constant $c_d = \frac{r_d}{r_{d-1}}$, where $r_n = \frac{2\pi^{\frac{d}{2}}}{\Gamma\left(\frac{d}{2}\right)}$ is the volume of the unit sphere embedded in Euclidean space $E^d$ (so $c_2=\pi$, $c_3=2$, $c_4=\frac{\pi}{2}$, $c_5 = \frac{4}{3}$, $c_6 = \frac{3}{8}\pi$, $\ldots$).  

Set:
\begin{equation}
k\{\vec{n}\}_x = \frac{1}{\sharp_x}\int_{B_\epsilon(x)}d\text{vol}\; g\; \vec{n},
\end{equation}
where $\sharp_x$ is the number of midpoints $\vec{n}\in B_\epsilon(x)$, the ball of radius $\epsilon$ about $x$.  Since $\epsilon \ll\kappa$, the $\vec{n}$ can be averaged (almost) unambiguously via parallel transport.

First define $j_0\left(v_x\right)$ on a tangent vector $v$ to $x$, by assigning a flow of velocity $||v_x||$ along every $n$ with needle midpoint $\hat{n}\in B_\epsilon(x)$ in the direction making the inner product $\langle v_x,\vec{n}\rangle > 0$.  (Note this rule assigns velocity $||v_x||$ not $\left|\langle v_x,\vec{n}\rangle\right|$ as one might guess.)  Now average $j_0$ to define $j$:

\begin{equation}
j(v) = \frac1{c_d}\int_{y\in B_x(\epsilon)} d\text{vol}\;g\;  j_0\left(v_y\right)
\end{equation}

Up to small errors, $k$ and $j$ are mutually inverse and (approximately) send vector fields $v$ with $\nabla\cdot v = 0$ to flow with zero combinatorial divergence, and vice versa. The factor $c_d$ compensates for the fact that a needle $n$ making angle $\theta$ with the smooth vector field $v_x$ at $x$ contributes only $\pm\cos\theta$ ($\operatorname{weight}(n)$) to the flux in direction $v_x$.

\bibliographystyle{JHEP}
\bibliography{refs}

\end{document}